\documentclass[pre, superscriptaddress, amsmath, amssymb, amsthm, preprint, nofootinbib, nobalancelastpage, tightenlines, floatfix]{revtex4-1}

\usepackage{graphicx}

\usepackage{array}%

\usepackage{color} 
\usepackage{esint}  
\usepackage{stmaryrd} 
\usepackage[nointegrals]{wasysym} 

\usepackage{xr}  

\usepackage{titlesec}

\newtheorem{theorem}{Theorem}
\newtheorem{corollary}[theorem]{Corollary}
\newtheorem{proof}{Proof}
\newtheorem{definition}{Definition}
\newtheorem{lemma}{Lemma}

\renewcommand{\thesection}{\arabic{section}}
\renewcommand{\thesubsection}{\thesection.\arabic{subsection}}

\renewcommand{\theequation}{\thesection.\arabic{equation}}

\titlespacing\section{0pt}{6pt plus 4pt minus 2pt}{6pt plus 2pt minus 2pt}
\titlespacing\subsection{0pt}{6pt plus 4pt minus 2pt}{6pt plus 2pt minus 2pt}
\titlespacing\subsubsection{0pt}{6pt plus 4pt minus 2pt}{6pt plus 2pt minus 2pt}

\newcommand{\wrap}{}

\newcommand{\comp}{\llcorner} 

\renewcommand{\vec}[1]{\boldsymbol {#1}}
\newcommand{\vect}[1]{\boldsymbol {#1}}
\newcommand{\tens}[1]{\boldsymbol{#1}}

\newcommand{\q}{\hat{p}}

\newcommand{\pos}{\vec{{q}}}
\newcommand{\mom}{\vec{\dot{q}}}
\newcommand{\sign}{\mathsf{sign}}

\newcommand{\f}{f}  
\newcommand{\hphi}{\vartheta}
\newcommand{\hf}{\hat{f}}

\newcommand{\x}{\vec{x}}
\renewcommand{\u}{\vec{u}}
\newcommand{\m}{\vec{m}}
\newcommand{\n}{\vec{n}}
\newcommand{\s}{\vec{s}}
\renewcommand{\r}{\vec{r}}

\newcommand{\valpha}{\vec{\alpha}}
\newcommand{\wa}{\mathring{\alpha}}
\renewcommand{\wp}{\mathring{p}}
\newcommand{\chemdens}{\psi}
\newcommand{\chempdf}{f}

\newcommand{\B}{\vec{B}}
\newcommand{\C}{\vec{C}}
\newcommand{\I}{\vec{I}}   

\newcommand{\U}{\vec{U}}
\newcommand{\V}{\vec{V}}
\newcommand{\W}{\vec{W}}
\newcommand{\X}{\vec{X}}
\newcommand{\Y}{\vec{Y}}

\newcommand{\0}{\vec{0}}
\newcommand{\1}{\vec{1}}
\newcommand{\2}{\vec{2}}

\newcommand{\D}{\mathcal{D}}
\newcommand{\R}{\mathbb{R}}
\newcommand{\N}{\mathbb{N}}

\newcommand{\Vop}{{\mathcal{V}}}
\newcommand{\Wop}{{\mathcal{W}}}

\newcommand{\vgrad}{\tens{G}}
\newcommand{\vgradtil}{\widetilde{\tens{G}}}
\newcommand{\J}{\tens{J}}
\newcommand{\M}{\tens{M}}
\newcommand{\K}{\tens{K}}

\newcommand{\F}{F}  

\newcommand{\vkappa}{\vec{\kappa}}
\newcommand{\hu}{\hat{\u}}
\newcommand{\dhu}{\dot{\hat{\u}}}
\newcommand{\vpartial}{\vec{\partial}}

\newcommand{\obs}{g}
\newcommand{\obss}{h}
\newcommand{\obsss}{\varsigma}

\newcommand{\lpos}{\llbracket}  
\newcommand{\rpos}{\rrbracket}  

\newcommand{\lgen}{\Lbag}  
\newcommand{\rgen}{\Rbag}  
\newcommand{\lgenn}{\llparenthesis}  
\newcommand{\rgenn}{\rrparenthesis}  

\newcommand{\delfour}{\blacktriangledown}  

\newcommand{\thn}{$\thorn$} 

\newcommand{\lvel}{\langle \hspace{-2pt} \langle \hspace{-2pt} \langle}  
\newcommand{\rvel}{\rangle \hspace{-2pt} \rangle \hspace{-2pt} \rangle} 
\newcommand{\lvol}{[ \hspace{-1pt} [ \hspace{-1pt}  [} 
\newcommand{\rvol}{] \hspace{-1pt} ] \hspace{-1pt}  ]}

\newcommand{\vminfive}{\vspace{-5pt}}
\newcommand{\vminten}{\vspace{-10pt}}
\newcommand{\vminfifteen}{\vspace{-20pt}}

\begin{document}

\title{Rethinking the Reynolds Transport Theorem, Liouville Equation, and Perron-Frobenius and Koopman Operators}


\author{Robert K. Niven}
\email{Email r.niven@adfa.edu.au}
\affiliation{School of Engineering and Information Technology, The University of New South Wales, Northcott Drive, Canberra ACT 2600, Australia.}
\affiliation{Institut Pprime (CNRS, Universit\'e de Poitiers, ISAE-ENSMA), Poitiers, France.}

\author{Laurent Cordier}
\affiliation{Institut Pprime (CNRS, Universit\'e de Poitiers, ISAE-ENSMA), Poitiers, France.}

\author{Eurika Kaiser}
\affiliation{Dept of Mechanical Engineering, University of Washington, Seattle, WA, USA.}

\author{Michael Schlegel}
\affiliation{Institut f\"ur Str\"omungsmechanik und Technische Akustik (ISTA),
Technische Universit\"at Berlin, 
D-10623 Berlin, Germany.}

\author{Bernd R. Noack}
\affiliation{Institute for Turbulence-Noise-Vibration Interaction and Control,
Harbin Institute of Technology, Shenzhen,
China}
\affiliation{Institut f\"ur Str\"omungsmechanik und Technische Akustik (ISTA),
Technische Universit\"at Berlin, 
D-10623 Berlin, Germany.}

\date{23 November 2020}%

\begin{abstract}
The Reynolds transport theorem provides a generalized conservation law for the transport of a conserved quantity by fluid flow through a continuous connected control volume. It is close connected to the Liouville equation for the conservation of a local probability density function, which in turn leads to the Perron-Frobenius and Koopman evolution operators. All of these tools can be interpreted as continuous temporal maps between fluid elements or domains, connected by the integral curves (pathlines) described by a velocity vector field. We here review these theorems and operators, to present a unified framework for their extension to maps in different spaces. These include (a) spatial maps between different positions in a time-independent flow, connected by a velocity gradient tensor field, and (b) parametric maps between different positions in a manifold, connected by a generalized tensor field. The general formulation invokes a multivariate extension of exterior calculus, and the concept of a probability differential form. The analyses reveal the existence of multivariate continuous (Lie) symmetries induced by a vector or tensor field associated with a conserved quantity, which are manifested as integral conservation laws in different spaces. The findings are used to derive generalized conservation laws, Liouville equations and operators for a number of fluid mechanical and dynamical systems, including spatial (time-independent) and spatiotemporal fluid flows, flow systems with pairwise or $n$-wise spatial correlations, phase space systems, Lagrangian flows, spectral flows, and systems with coupled chemical reaction and flow processes. 
\end{abstract}



\maketitle


\section{\label{Intro}Introduction}

In the early 20th century, building on his successes in the analysis of fluid turbulence \cite{Reynolds_1883, Reynolds_1895}, Osborne Reynolds presented what is now called the {\it Reynolds transport theorem}, a generalized equation for the transport of a conserved quantity by fluid flow through a stationary or moving continuous control volume \cite{Reynolds_1903}. This reduces in particular circumstances to the integral and differential conservation laws (such as for mass, momentum and energy) of fluid mechanics.  The Reynolds transport theorem and its subsidiary conservation equations -- with their paradigm of an Eulerian velocity field -- as well as Reynolds' insights into fluid turbulence, now provide the foundation for the overwhelming proportion of theoretical and numerical models used by practitioners in fluid mechanics.  

In the older field of classical mechanics, Liouville \cite{Liouville_1838} presented a relation for the derivative of a state function which, when later applied in statistical mechanics, gives a conservation equation for the local probability density function in time \cite{Liouville_1838, Lutzen_1990}. 
While often grouped with the Fokker-Planck equation \cite[e.g.,][]{Risken_1984}, the latter includes the effect of stochastic processes or diffusion.  In the early 20th century, developments in matrix theory \cite{Perron_1907, Frobenius_1912} led to the Perron-Frobenius (or Ruelle-Perron-Frobenius) operator \cite{Ruelle_1968} and its dual Koopman operator \cite{Koopman_1931, Koopman_vonNeumann_1932}, for extrapolation of a time-evolving density or observable, respectively, from an initial value. These operators have the advantage of linearity, enabling the conversion of a nonlinear dynamical system into a linear evolution equation, {albeit at the expense of infinite dimensionality of the operator}. Over the past decade, there has been considerable interest in the theory and application of these operators to a variety of dynamical and fluid flow systems \cite[e.g.,][]{Mezic_2005, Rowley_etal_2009, Chen_etal_2012, Budisic_etal_2012, Bagheri_2013, Mezic_2013, Bagheri_2014}.


\setlength\tabcolsep{3pt}

\begin{figure*}[t] \footnotesize \sffamily
\makebox[\textwidth][c]{
\begin{tabular*}{522pt}{| m{55pt} | m{130pt} | m{130pt} | m{180pt} |}  
\hline
\vminfifteen
\small{Space}
\vminfifteen
&
\vminfifteen
\begin{gather*}
\text{\small{Geometric space $\Omega \subset \R^3$}}
\end{gather*}
\vminfifteen
&
\vminfifteen
\begin{gather*}
\text{\small{Velocity domain $\D \subset \R^3$}}
\end{gather*}
\vminfifteen
&
\vminfifteen
\begin{gather*}
\text{\small{Submanifold $\Omega^n$ $\subset$ manifold $M^n$}} 
\end{gather*}
\vminfifteen
  \\
  \hline
\vminfifteen
\small{Coordinates, parameters}
\vminfifteen
&
\vminfifteen
\begin{gather*}
\text{\small{Geometric coordinates $\x$, }}
\\
\text{\small{temporal map $t$}}
\end{gather*}
  \vminfifteen
&  \vminfifteen
\begin{gather*}
\text{\small{Velocity coordinates $\u$, }}
\\
\text{\small{spatial map $\x$}}
\end{gather*}
  \vminfifteen
  &
    \vminfifteen
\begin{gather*}
\text{\small{Local coordinates $\X$, }} 
\\
\text{\small{parametric map $\C$}} 
\end{gather*}
  \vminfifteen
  \\
  \hline
\small{Vector or} \small{tensor field}
&
  \vminfifteen
  \begin{gather*} \text{\small{Velocity vector field }} \u_{rel}(\x,t) \end{gather*}
  \vminfifteen
&  \vminfifteen
  \begin{gather*} \text{\small{Velocity gradient tensor field }} \\ \vgrad_{rel}(\u,\x) \end{gather*}
  \vminfifteen
  &
    \vminfifteen
  \begin{gather*} \text{\small{Tensor field }} \V \end{gather*}
  \vminfifteen
  \\
  \hline
\small{Density}
&
  \vminfifteen
\begin{gather*}
\text{\small{Volumetric density }}\\
\alpha(\x,t) 
\end{gather*}
  \vminfifteen
&  \vminfifteen
\begin{gather*}
\text{\small{Phase space density }} \\
\varphi(\u, \x) 
\end{gather*}
  \vminfifteen
  &
    \vminfifteen
\begin{gather*}
\text{\small{Field of differential $r$-forms }}
\omega^r 
\end{gather*}
  \vminfifteen
  \\
  \hline
\small{Reynolds transport theorem}
&
 \begin{picture}(120,110)
 \put(12,0){ \includegraphics[height=108pt]{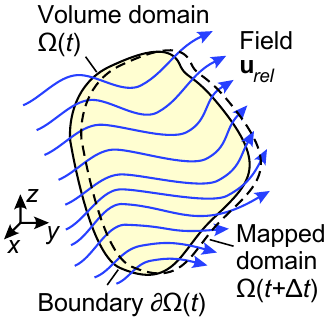}}
 \end{picture}
  \vminfive
 \begin{gather*}
  \dfrac{d}{dt} \iiint\limits_{\Omega(t)} \alpha \, d^3 \x 
=   \dfrac{D}{Dt} \iiint\limits_{\Omega(t)} \alpha \, d^3 \x 
\\ 
=   \iiint\limits_{\Omega(t)}  \biggl[ \dfrac{\partial \alpha}{\partial t}   +  
\nabla_{\x} \cdot (\alpha \, \u_{rel} ) \biggr ] d^3 \x 
\end{gather*}
 \vminfive
&
 \begin{picture}(120,110)
\put(12,0){ \includegraphics[height=108pt]{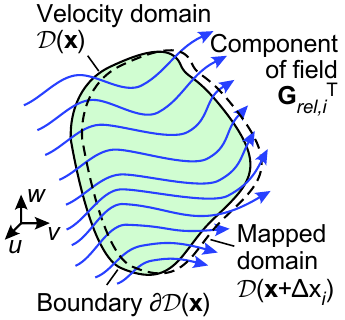} }
 \end{picture}
  \vminfive
\begin{gather*}
d \iiint\limits_{\D(\x)} \varphi \, d^3 \u 
 =
\biggl[
\iiint\limits_{\D(\x )}
\bigl(
\nabla_{\x} \varphi
+
\\
\nabla_{\u} \cdot ( \varphi \, \vgrad_{rel}^\top )
\bigr) d^3 \u
\biggr] \cdot
d \x
\end{gather*}
 \vminfive
&
  \begin{picture}(130,110)
\put(12,-3){ \includegraphics[height=108pt]{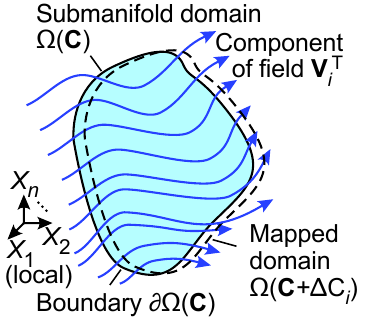} }
 \end{picture}
  \vminfive
 \begin{multline*}
\hat{d} \int\limits_{\Omega(\C)}  \omega^r 
=  \biggl[ \int\limits_{\Omega(\C)} \mathcal{L}_{\V \comp \C}^{(\C)}  \, \omega^r \biggr] \cdot {d\C} 
\\
=   \biggl[  \int\limits_{\Omega(\C)}  
\vpartial_{\C} \omega^r +  i_{\V}^{(\C)}  \, {d} \omega^r + {d} ( i_{\V}^{(\C)} \, \omega^r  )
\biggr] \cdot  {d\C} 
\end{multline*}
  \vminfive
 \\
 \hline
 \small{Probability} &
 \vminfifteen
  \begin{gather*}
 p(\x|t) d^3\x
 \end{gather*}
 \vminfifteen
&
  \vminfifteen
\begin{gather*}
f(\u|\x) d^3\u
\end{gather*}
 \vminfifteen
&
  \vminfifteen
  \begin{gather*}
\text{\small{Field of probability $r$-forms }} \rho^r 
\\
\text{\small{ conditioned on }} \C
\end{gather*}
 \vminfifteen
\\
 \hline
 \small{Liouville equation}
 &
 \vminten
 \begin{gather*}
 \dfrac{\partial p}{\partial t}   +  \nabla_{\x} \cdot (p \, \u_{rel} )=0
  \end{gather*}
 \vminten
&
 \vminten
 \begin{gather*}
\nabla_{\x} f  +  
 \nabla_{\u} \, \cdot ( f \, \vgrad_{rel}^\top ) = \vect{0}
  \end{gather*}
 \vminten
&
 \vminten
\begin{gather*}
 \mathcal{L}_{\V \comp \C}^{(\C)} \, \rho^r
= \vec{0}
\end{gather*}
 \vminten
 \\
 \hline
 \small{Observable} &
 \vminfifteen
  \begin{gather*}
 \obs(\x,t) 
 \end{gather*}
 \vminfifteen
&
  \vminfifteen
\begin{gather*}
\obss(\u,\x) 
\end{gather*}
 \vminfifteen
&
  \vminfifteen
  \begin{gather*}
 \text{\small{Field of differential 0-forms }} \obsss
\end{gather*}
 \vminfifteen
\\
\hline
 \small{Expected value} &
 \vminfifteen
  \begin{gather*}
 \text{\small{Volumetric mean}}
 \\
\lvol {\obs} \rvol (t) =\iiint\limits_{\Omega(t)} \obs \, p \, d^3 \x  
 \end{gather*}
 \vminfive
 &
  \vminfifteen
\begin{gather*}
 \text{\small{Ensemble mean}}
 \\
\lvel{\obss} \rvel (\x) = \iiint\limits_{\D(\x)} \obss \, f \, d^3 \u  
\end{gather*}
 \vminfive
&
  \vminfifteen
  \begin{gather*}
 \text{\small{Submanifold mean}}
 \\
\lgen {\obsss} \rgen (\C)  = \int\limits_{\Omega(\C)} \obsss \, \rho^r 
\end{gather*}
 \vminfive
\\
\hline
\small{Evolution operators} &
 \vminfifteen
 \begin{gather*}
p(\x|t)=\hat{P}_t \, p(\x|0)
\\
\obs(\x,t)=\hat{K}_t \, \obs(\x,0)
 \end{gather*}
 \vminfifteen
&
 \vminfifteen
 \begin{gather*}
f(\u|\x)=\hat{{P}}_{\x} \, f(\u| \vec{0} )
\\
\obss(\u,{\x})=\hat{{K}}_{\x} \, \obss(\u,\vec{0})
 \end{gather*}
 \vminfifteen
&
\vminfifteen
 \begin{gather*}
\rho^r_{\C}=\hat{P}_{\C} \, \rho^r_{\0}
\\
\obsss_{\C}=\hat{{K}}_{\C} \, \obsss_{\0}
 \end{gather*}
 \vminfifteen
\\
\hline
\end{tabular*} 
}
\caption{Summary of the main formulations presented in this study (for definitions of symbols, see text).}
\label{fig:scope}
\end{figure*}


Despite more than a century of mathematical generalization to different fields and spaces, including of the major integral theorems of vector calculus (the gradient, divergence and Stokes' theorems), most presentations of the Reynolds transport theorem and Liouville equation are expressed in terms of the time evolution of the density of a conserved quantity carried by a three-dimensional velocity field. However, some hints have emerged of more general formulations. The equivalence of conservation laws and symmetries has been appreciated since the famous works of Lie \cite{Lie_Engel_1888, Lie_1891} and Noether \cite{Noether_1918}, and multiparameter Lie groups and other generalizations have been invoked by some authors \cite[e.g.,][]{Bluman_Kumei_1996, Baumann_2000, Oliveri_2010}. Recently, there has been new interest in the rescaling of fluid flow equations using one-parameter Lie transformations, including of the Reynolds transport theorem, Navier-Stokes and Reynolds-averaged Navier-Stokes equations  \cite{Haltas_Ulusoy_2015, Ercan_Kavvas_2015, Ercan_Kavvas_2017}. 
Furthermore, Sharma and co-workers \cite{Sharma_etal_2016} introduced spatial and spatiotemporal Koopman operators for the analysis of turbulent flow systems, to exploit underlying symmetries (coherent structures) evident in the Navier-Stokes equations. These new formulations and their connections to singular value decomposition (SVD) and dynamic mode decomposition (DMD) are now the subject of intense scrutiny in the literature \cite[e.g.,][]{Hemati_etal_2017, Arbabi_Mezic_2017, Proctor_etal_2018, Kutz_etal_2018, LeClainche_Vega_2018a, LeClainche_Vega_2018b, Giannakis_etal_2019, Froyland_etal_2020, Umbria_etal_2020, Giannakis_Das_2020, Perez_etal_2020, Bai_etal_2020}.  

Many years ago, Flanders \cite{Flanders_1973} viewed the Reynolds transport theorem as not merely a theorem of fluid mechanics, but a three-dimensional generalization of the Leibniz rule for differentiation of an integral. If so, it is far more general and powerful than its current usage might suggest.  Flanders then extended the theorem to the flow of an $r$-dimensional compact submanifold within an $n$-dimensional manifold, expressed using the formalism of exterior calculus \cite{Flanders_1973, Lee_2009, Frankel_2013}. Recently this was generalized to include the analysis of evolving cycles or differential chains \cite{Harrison_2015}, and thereby to fixed and evolving irregular domains on a manifold \cite{Seguin_Fried_2014, Seguin_etal_2014, Falach_Segev_2014}. Several authors have also reported a {\it surface transport theorem}, a two-dimensional analog of the Reynolds transport theorem \cite{Gurtin_etal_1989, Ochoa-Tapia_etal_1993, Slattery_Sagis_Oh_2007, Fosdick_Tang_2009, Lidstrom_2011, Seguin_etal_2014}. However, all these formulations still only provide a one-parameter (temporal) map induced by a stationary or time-evolving velocity vector field.  

Separately, a number of researchers have presented a spatial variant of the Reynolds transport theorem, termed the {\it spatial averaging theorem}, based on spatial rather than time derivatives of volumetric integrals. This theorem connects the volume average of a gradient (or divergence) to the gradient (or divergence) of a volume average, with important applications to flows in porous media and multiphase flows. The theorem has been presented in several forms \cite{Anderson_Jackson_1967, Whitaker_1967a, Whitaker_1967b, Slattery_1967, Marle_1967, Whitaker_1969, Slattery_1972, Bachmat_1972, Whitaker_1973, Gray_Lee_1977} 
and is the subject of various proofs \cite{Whitaker_1967b, Marle_1967,  Whitaker_1969, Slattery_1972, Bachmat_1972, Whitaker_1973, Gray_Lee_1977, Cushman_1982, Cushman_1983, Howes_Whitaker_1985, Whitaker_1985, Chen_1994, Whitaker_1999, Slattery_1999}.  It has also been generalized to give a variety of averaging relations in geometric space and time \cite{Drew_1971, Bachmat_1972, Whitaker_1973, Gray_Lee_1977, Cushman_1982, Cushman_1983, Cushman_1984, Whitaker_1986, Bachmat_Bear_1986, Plumb_Whitaker_1988, Bear_Bachmat_1991, Ochoa-Tapia_etal_1993, Gray_etal_1993, Quintard_Whitaker_1994a, Grau_Cantero_1994, He_Sykes_1996, Whitaker_1999, Slattery_1999, Davit_etal_2010, Gray_Miller_2013, Wood_2013, Pokrajac_deLemos_2015, Takatsu_2017}, but not, it appears, to other spaces. 

In this review, we first explore (\S\ref{sect:temp_anal}) the volumetric-temporal formulation of the Reynolds transport theorem, and its lesser-known connection to the Liouville equation and the Perron-Frobenius and Koopman operators. These insights are then applied to develop a unified framework for the derivation of these theorems and operators in different spaces, firstly in \S\ref{sect:spat_anal} to give their velocimetric-spatial analogs, and then in \S\ref{sect:general} to provide more general parametric formulations. The most general formulations invoke multivariate extensions of several operators of exterior calculus, including the Lie derivative, and the concept of a probability differential form. The formulations presented herein are summarized in Figure \ref{fig:scope}, and are supported by mathematical proofs given in the main text and Appendices. The breadth of the findings are then demonstrated in \S\ref{sect:ex_sys} by application to a variety of flow and dynamical systems, including spatial (time-independent) and spatiotemporal fluid flows, flow systems with pairwise or $n$-wise spatial correlations, phase space systems, Lagrangian flows, spectral flows, and systems with coupled chemical reaction and flow processes. Our concluding comments are given in \S\ref{sect:concl}.

\section{\label{sect:temp_anal}Temporal Analyses} 
\subsection{\label{sect:temp_Re_tr} Volumetric-Temporal Reynolds Transport Theorem} 
We first revisit an extended form of the standard or {\it volumetric-temporal} Reynolds transport theorem \cite{Reynolds_1903}, where the two adjectives refer respectively to the domain of integration and the parameter, which becomes the differentiation variable. For later generality, we use a slightly different notation to that commonly used in fluid mechanics.  
\begin{theorem}
Consider a continuum represented by the Eulerian description, in which each local property of a fluid can be specified as a function of Cartesian position coordinates $\x=[x,y,z]^\top$ and time $t$ as the fluid moves past. 
Let $\alpha(\x, t)$ be the concentration or density of a conserved quantity (scalar, vector or tensor) within the fluid, expressed per unit volume. 
{Let $\alpha(\x, t)$ be continuous and continuously differentiable in space and in time throughout a moving body of fluid (the ``fluid volume'', ``material volume'' or ``domain'') $\Omega(t)$, for all positions up to the boundary and all times considered. The total derivative of the conserved quantity within the fluid volume} as it moves through an enclosed, moving, smoothly deformable region of space (the ``control volume'') satisfies \cite{Truesdell_Toupin_1960, Slattery_1972, Bear_Bachmat_1991, White_1986, Dvorkin_Goldschmidt_2005, Munson_etal_2010, Lidstrom_2011}:
\begin{gather}
\begin{split}
\frac{d}{dt} \iiint\limits_{\Omega(t)} \alpha \, d^3 \x 
&=    \iiint\limits_{\Omega(t)}  \frac{\partial \alpha}{\partial t}  \, d^3 \x + \oiint\limits_{\partial \Omega(t)} \alpha \, \u_{rel} \cdot d^2 \x  
\wrap  
=   \iiint\limits_{\Omega(t)}  \biggl[ \frac{\partial \alpha}{\partial t}   +  \nabla_{\x} \cdot (\alpha \, \u_{rel} ) \biggr ] d^3 \x,
\end{split}
\label{eq:Re_tr}
\end{gather}
where 
$\partial \Omega(t)$ is the domain boundary, $\u_{rel}(\x,t)$ is the velocity of the fluid relative to the control volume, $d/dt$ is the total derivative (here equivalent to the material or substantial derivative, often written $D/Dt$), $\partial/ \partial t$ is the derivative at fixed position, $\nabla_{\x} = \partial/\partial [x, y, z]^\top$ is the nabla operator with respect to $\x$,  $d^3 \x = dV= dxdydz$ is an infinitesimal volume element in the domain, and $d^2 \x = d\vect{A} = \vect{n} dA$ is an infinitesimal directed area element at the boundary, where $\vect{n}$ is the outward unit normal. 
\end{theorem}

\begin{proof}
Proofs of \eqref{eq:Re_tr}, for either the extended form given here or for the simpler case of a stationary control volume (see discussion below), have been given using the tools of continuum mechanics \cite{Reynolds_1903, Prager_1961, Sokolnikoff_Redheffer_1966, White_1986, Bear_Bachmat_1991, Tai_1992, Leal_2007, Munson_etal_2010}, Lagrangian coordinate transformation \cite{Aris_1962, Slattery_1972, Flanders_1973, Spurk_1997, Dvorkin_Goldschmidt_2005} and exterior calculus \cite{Flanders_1973, Frankel_2013, Lee_2009, Lidstrom_2011}.  Variants of the first two proofs of \eqref{eq:Re_tr} are given in  \ref{sect:Apx_Re_temp}. 
An extended exterior calculus formulation is also given in \S\ref{sect:general}, and shown to reduce to a generalized vector calculus formulation. 
\end{proof}

We note that \eqref{eq:Re_tr} is a special case of the Helmholtz transport theorem, involving surfaces that are not closed \cite{Tai_1992}. Furthermore, extensions of \eqref{eq:Re_tr} have been derived for fluids with fixed or moving discontinuities in $\alpha(\x,t)$ and/or in $\u(\x,t)$, requiring additional integral terms \citep[e.g.][]{Truesdell_Toupin_1960, Dvorkin_Goldschmidt_2005, Myers_2015}. As mentioned earlier, extensions of the Reynolds transport theorem have also been presented for evolving irregular domains and rough surfaces \cite{Seguin_Fried_2014, Seguin_etal_2014, Falach_Segev_2014}.


Examples of the conserved quantity $\alpha(\x,t)$ commonly used in the Reynolds transport theorem \eqref{eq:Re_tr} include the fluid mass density $\rho$; the mass density (concentration) $\rho_c$ of a chemical species $c$; the linear momentum $\rho \u$, the angular momentum $\rho(\r \times \u)$, where $\r$ is the radius of a local lever arm, the energy density $\rho e$, where $e$ is the specific energy; the charge density $\rho z$, where $z$ is the specific charge; and the entropy density $\rho s$, where $s$ is the specific entropy \cite{Aris_1962, White_1986, Spurk_1997, Leal_2007, Munson_etal_2010}.  In standard applications, the left-hand term $\frac{d}{dt} \iiint\nolimits_{\Omega(t)} \alpha \, d^3 \x$ of the Reynolds transport theorem is then used to capture any non-zero sources or sinks of the conserved quantity represented by $\alpha$. For the examples given these include, respectively, the rate of production of fluid mass $Dm_f/dt$ in the fluid volume (usually taken as zero); the rate of production $Dm_c/dt$ of the mass of species $c$ due to chemical reaction in the fluid volume; the total force on the fluid volume $\sum \vec{F}_{FV}$; the total torque on the fluid volume $\sum \vec{T}_{FV}$; the sum of heat and work flows $(\dot{Q}_{in} + \dot{W}_{in})$ into the fluid volume; the total electric current $I_{FV}$ into the fluid volume; and the sum of the entropy production and non-fluid entropy flow rate $(\dot{\sigma} + \dot{S}^{nf}_{FV})$ into the fluid volume \cite{Aris_1962, White_1986, Spurk_1997, Leal_2007, Munson_etal_2010}. 

In \eqref{eq:Re_tr}, we must carefully consider the meaning of the relative velocity $\u_{rel}$. In the surface integral form, $\u_{rel}$ expresses the velocity of the fluid relative to the control volume at the boundary, in some references described as the velocity of the fluid surface $\partial \Omega(t)$ \cite[e.g.,][]{Truesdell_Toupin_1960, Dvorkin_Goldschmidt_2005, Lidstrom_2011, Takatsu_2017}.  For Cartesian velocity coordinates, this can be identified as $\u_{rel} = \u - \u_{CV}$, where $\u_{CV}$ is the velocity of the control volume and $\u$ is the intrinsic velocity of the fluid \cite{Reynolds_1903, White_1986, Munson_etal_2010} (see analysis in  \ref{sect:Apx_Re_temp}). 
In consequence $\u_{rel} \cdot \vec{n}$ gives the volumetric flux normal to and out of the control surface.
In the volumetric integral form, $\u_{rel}$ expresses the velocity of any point in the fluid relative to the moving control volume. The latter thus invokes -- by the Gauss-Ostrogradsky divergence theorem -- the existence of a continuous and continuously differentiable vector field $\u_{rel}$, which by continuity must extend throughout the entire space in which the fluid is present. 
For consistency, the total or substantial derivative should be defined with respect to this moving frame of reference \cite{Slattery_1972, Truesdell_Toupin_1960, Dvorkin_Goldschmidt_2005}:
\begin{align}
\frac{d \alpha}{dt} =\frac{D \alpha}{Dt} :=\frac{\partial \alpha}{\partial t} + \nabla_{\x} \alpha \cdot \u_{rel}
\label{eq:total_deriv}
\end{align}
(see discussion in  \ref{sect:Apx_Re_temp}). Combining \eqref{eq:total_deriv} and the final form of \eqref{eq:Re_tr} gives a total derivative form of the Reynolds transport theorem \cite[e.g.,][]{Aris_1962, Slattery_1972, Dvorkin_Goldschmidt_2005}:
\begin{align}
\frac{d}{dt} \iiint\limits_{\Omega(t)} \alpha \, d^3 \x 
&=   \iiint\limits_{\Omega(t)}  \biggl[ \frac{d \alpha}{d t}   +  \alpha \nabla_{\x} \cdot  \u_{rel}   \biggr ] d^3 \x.  
\label{eq:Re_tr5}
\end{align}
For a stationary control volume $\u_{CV}=\0$, \eqref{eq:Re_tr} and \eqref{eq:Re_tr5} reduce to intrinsic forms of the Reynolds transport theorem, based on the intrinsic velocity field $\u$. For both a stationary control volume $\u_{CV}=\0$ and a stationary fluid $\u=\0$, the surface integral term (or equivalently, the divergence term) in \eqref{eq:Re_tr} vanishes.




As mentioned, some authors have reported a {\it surface transport theorem}, a two-dimensional analog of the Reynolds transport theorem \eqref{eq:Re_tr} for the total derivative of the surface integral of a surface density \cite{Seguin_etal_2014, Gurtin_etal_1989, Ochoa-Tapia_etal_1993, Slattery_Sagis_Oh_2007, Fosdick_Tang_2009, Lidstrom_2011}. This contains surface and line integral terms, which can be  combined using a surface divergence theorem \cite{Ochoa-Tapia_etal_1993, Slattery_Sagis_Oh_2007, Lidstrom_2011, Gray_etal_1993}. This theorem has been proven by differential calculus methods \cite{Ochoa-Tapia_etal_1993, Gurtin_etal_1989, Fosdick_Tang_2009, Gray_etal_1993} and Lagrangian coordinate transformation \cite{Slattery_Sagis_Oh_2007}, analogous or related to the proofs given in  \ref{sect:Apx_Re_temp}, and also by exterior calculus methods \cite{Seguin_etal_2014, Lidstrom_2011}. 

\subsection{\label{sect:temp_prob} Probabilistic Analysis and the Temporal Liouville Equation} 

The connections between the Reynolds transport theorem and Liouville equation are not widely known, but are reported by some authors \cite[e.g.,][]{Ehrendorfer_2003}. Consider a fluid flow system with the observables described using a multivariate random variable for position $\vect{\Upsilon}_{\x} = [\Upsilon_x, \Upsilon_y, \Upsilon_z]^\top$ with values $\x$, and a random variable for time $\Upsilon_t$ with values $t$\footnote{We note that  random variables of observable quantities are commonly denoted by corresponding capital letters \cite[e.g.,][]{Cover_T_2006}. Due to clashes with standard symbols, we use a different notation.}.
We then define the joint-conditional probability density function (pdf) $p(\x|t)$ over the domain $\Omega(t)$, to represent the probability that at the time infinitesimally close to the specified time $t$, a fluid element will be infinitesimally close to the position $\x$:
\begin{equation}
\begin{split}
p (\x|t) \, d^3 \x = p (x,y,z|t) \, dxdydz
\approx
\text{Prob}
&\Biggl(
\begin{matrix} 
x \le \Upsilon_x \le x+dx \\
y \le \Upsilon_y \le y+dy \\
z \le \Upsilon_z \le z+dz 
\end{matrix} 
\Biggl | \;
t \le \Upsilon_t \le t+dt
\Biggr).
\end{split}
\label{eq:pdf_xvect_sub_t}
\end{equation}
{(Formally, the pdf $\protect{p(\x|t)}$ is defined over continuous intervals of space and time, from which \eqref{eq:pdf_xvect_sub_t} applies in the infinitesimal limits \cite[e.g.,][]{Feller_1966}.)}
The pdf will satisfy normalization for any time $t$:
\begin{align}
1 = \iiint\limits_{\Omega(t)} \, p (\x|t) \, d^3 \x.
\label{eq:local_norm_x}
\end{align}
We also define the time-dependent volumetric average of an observable $\obs(\x, t)$:
\begin{align}
\lvol{\obs} \rvol (t) 
= \iiint\limits_{\Omega(t)} \obs(\x, t) \, p (\x|t) \, d^3 \x. 
\label{eq:vol_average}
\end{align}
Now if $p(\x | t)$ satisfies the same regularity conditions as the density $\alpha(\x,t)$ in \S\ref{sect:temp_anal}, we can substitute $\alpha(\x,t) = p(\x|t)$ into the volumetric-temporal Reynolds transport theorem \eqref{eq:Re_tr}, to directly give:
\begin{align}
\begin{split}
\frac{d}{dt}  \iiint\limits_{\Omega(t)} p \, d^3 \x 
&=   \iiint\limits_{\Omega(t)}  \biggl[ 
\frac{\partial p}{\partial t}   +  \nabla_{\x} \cdot (p \, \u_{rel} ) \biggr ] d^3 \x,
\end{split}
\label{eq:Re_tr_x_prob}
\end{align}
However, from \eqref{eq:local_norm_x}, the left-hand side of \eqref{eq:Re_tr_x_prob} vanishes for all $t$. 
This leads to the following theorem:
\begin{theorem}
Let $p(\x|t)$ be the probability density of the position $\x$ at the specified time $t$, defined over a fluid volume $\Omega(t)$ containing a relative velocity vector field $\u_{rel}(\x,t)$. Let $p(\x|t)$ be continuous and continuously differentiable in space and in time throughout $\Omega(t)$, for all positions up to the boundary and all times considered. For all $\x \in \Omega(t)$ and all times $t$:
\begin{align}
\frac{\partial p}{\partial t}   +  \nabla_{\x} \cdot (p \, \u_{rel} )=0.
\label{eq:Liouville}
\end{align}
\end{theorem}
\begin{proof}
Eq.\ \eqref{eq:Liouville} follows directly from \eqref{eq:local_norm_x}, \eqref{eq:Re_tr_x_prob} and the fundamental lemma of the calculus of variations \cite{Weinstock_1952, Gelfand_Fomin_1963}, for all choices of compactly supported continuous and continuously differentiable pdfs $p(\x|t)$.  
\end{proof}

%
Eq.\ \eqref{eq:Liouville} is known as the standard or temporal {\it Liouville equation} for a fluid flow system, for conservation of the conditional pdf $p(\x|t)$ under the relative velocity $\u_{rel}$. 
We emphasize that the above proof does not apply to discontinuous or non-differentiable $p(\x|t)$, and important exceptions may occur, e.g., in a shock wave or mixing layer. 
%

\subsection{\label{sect:temp_ops} Temporal Perron-Frobenius and Koopman Operators}

Taking the analysis farther, the solution to \eqref{eq:Liouville} can be written as the probabilistic evolution equation $p(\x|t)=\hat{P}_t \, p(\x|0)$, where $\hat{P}_t$ is the Perron-Frobenius operator, with origin $t=0$ measured in the relative coordinate system of $t$ \cite[e.g.,][]{Froyland_Padberg_2009, Kaiser_etal_2014}. Examining a probability product, it is readily verified that this is linear, giving $\hat{P}_t = \exp(t \, \hat{L}_t)$, in which $\hat{L}_t$ is the (multiplicative) temporal Liouville operator defined by $\hat{L}_t p = -\nabla_{\x} \cdot (p \, \u_{rel} )$. 
The Koopman operator $\hat{K}_t$ adjoint to $\hat{P}_t$ can then be defined from the volume average \eqref{eq:vol_average}, based on the duality:
\begin{align}
\begin{split}
\lvol{\obs} \rvol (t) 
&= \iiint\limits_{\Omega(t)} \obs(\x, t) \, \hat{P}_t \, p(\x|0) \, d^3 \x  
\wrap
= \iiint\limits_{\Omega(t)} \hat{K}_t \, \obs(\x,0) \, p (\x|t) \, d^3 \x.
\end{split}
\label{eq:op_duality}
\end{align}
The Koopman operator provides an evolution equation for the observable $\obs(\x,t)=\hat{K}_t \, \obs(\x,0)$, and can be determined by spectral decomposition, with close connections to SVD and DMD \cite[e.g.,][]{Mezic_2005, Rowley_etal_2009, Chen_etal_2012, Budisic_etal_2012, Bagheri_2013, Mezic_2013, Bagheri_2014, Hemati_etal_2017, Arbabi_Mezic_2017, Proctor_etal_2018, Kutz_etal_2018, LeClainche_Vega_2018a, LeClainche_Vega_2018b, Giannakis_etal_2019, Froyland_etal_2020, Umbria_etal_2020, Giannakis_Das_2020, Perez_etal_2020, Bai_etal_2020}. 

\subsection{\label{sect:temp_other} Further Simplifications}

{\bf Intrinsic Flows}: For a stationary frame of reference $\u_{CV}$ $=\0$, we recover the intrinsic Liouville equation \cite{Pottier_2010}:
\begin{align}
\frac{\partial p}{\partial t}   +  \nabla_{\x} \cdot (p \, \u )=0.
\label{eq:Liouville_intr}
\end{align}
This expresses the local conservation of $p$ under its intrinsic motion \cite{Pottier_2010}. 
This can be compared to the Fokker-Planck equation \cite{Risken_1984}:
\begin{align}
\frac{\partial p}{\partial t}   +  \nabla_{\x} \cdot (p \, \u) -  \nabla^2_{\x} : (\tens{D} \, p )=0,
\label{eq:Fokker-Planck}
\end{align}
in which $\nabla_{\x}^2=\nabla_{\x} (\nabla_{\x})^\top$ is the second derivative or Hessian operator, $\tens{D}$ is a diffusion tensor and ``$:$'' is the tensor scalar product. Evidently, the Fokker-Planck equation is inconsistent with Reynold's transport theorem \eqref{eq:Re_tr}, and contains a pdf which is not conserved locally.  The distinction lies in the fact that in \eqref{eq:Liouville_intr}, the pdf $p(\x|t)$ is considered to extend over the entire domain $\Omega(t)$, whereas in \eqref{eq:Fokker-Planck} it also undergoes diffusion into previously unoccupied regions.

From dynamical systems theory, we can consider \eqref{eq:Liouville_intr} to be induced by 
${d\x}/{dt} = \u = \tens{\mathcal{F}} (\x)$,
where $\tens{\mathcal{F}}$ is the (vector) propagator \cite{Gaspard1995pre}.
For incompressible or solenoidal flow $\nabla_{\x} \cdot \u = 0$, \eqref{eq:Liouville_intr} simplifies further to give the more common total derivative form:
\begin{align}
\frac{\partial p}{\partial t}   +  \nabla_{\x} p \cdot \u     = \frac{d p}{d t}  =0.
\end{align}


\noindent {\bf Two-Dimensional Flows}: Alternatively, consider the special case of two-dimensional flow with position $\x=[x,y]^\top$ and relative velocity $\u_{rel}= [u_{rel},v_{rel}]^\top $. 
Defining the stream function $\Psi$ by the relations
$u_{rel} = {\partial \Psi}/{\partial y}$ and $v_{rel} = - {\partial \Psi}/{\partial x}$
\cite[c.f.,][]{Prager_1961, Munson_etal_2010}, 
substitution in the general Liouville equation \eqref{eq:Liouville}, using the relative solenoidal condition $\nabla_{\x} \cdot \u_{rel}=0$, gives the Hamiltonian-like form:
\begin{align}
\frac{\partial p}{\partial t}   + \biggl( \frac{\partial \Psi}{\partial y} \frac{\partial p}{\partial x} - \frac{\partial \Psi}{\partial x}   \frac{\partial p}{\partial y} \biggr) = 0.
\label{eq:Liouville2D}
\end{align}
This definition allows for a moving and smoothly deforming control volume.
The stream function $\Psi$ is normal to the relative velocity potential $\Phi$, defined for irrotational relative flow $\nabla_{\x} \times \u_{rel}=0$ by $\u_{rel}=\nabla_{\x} \Phi$, hence 
$u_{rel} = {\partial \Phi}/{\partial x}$ and $v_{rel} =  {\partial \Phi}/{\partial y}$
\cite{Prager_1961, Munson_etal_2010}. 
For steady flows, these give a {\it flow net} of curvilinear orthogonal coordinates $(\Psi,\Phi)$, tangential and normal to the relative velocity vector.


\section{\label{sect:spat_anal}Three-Dimensional Spatial Analyses} 

\subsection{\label{sect:spat_av_thm} Volumetric-Spatial Reynolds Transport Theorem (Spatial Averaging Theorem)}

As mentioned in the introduction, several researchers have developed a spatial analog of the Reynolds transport theorem \eqref{eq:Re_tr}, involving spatial rather than time derivatives of volumetric averages. A one-dimensional form of this theorem can be written as \cite{Slattery_1967, Whitaker_1969, Slattery_1972, Takatsu_2017}:
\begin{gather}
\frac{d}{ds} \iiint\limits_{\Omega_f(t)} \alpha \, d^3 \x 
=    \iiint\limits_{\Omega_f(t)}  \frac{\partial \alpha}{\partial s}  \, d^3 \x + \oiint\limits_{\partial \Omega_f(t)} \alpha \, \frac{d\x}{ds} \cdot d^2 \x  
\label{eq:spatial_av_1D}
\end{gather}
where $\alpha(\x(s),t)$  again is the generalized volumetric density, $s$ is an intrinsic spatial coordinate, $\Omega_f(t)$ is the fluid volume and $\partial \Omega_f(t)$ is the fluid surface in a multiphase system (including interior surfaces). All other quantities are as defined in \S\ref{sect:temp_anal}. For the $\imath$th phase, this gives the following gradient and divergence spatial averaging theorems \cite{Whitaker_1967a, Whitaker_1967b, Marle_1967, Whitaker_1969, Slattery_1972, Bachmat_1972, Whitaker_1973, Gray_Lee_1977, Cushman_1982, Howes_Whitaker_1985, Chen_1994, Whitaker_1999, Slattery_1999}:
\begin{gather}
\begin{split}
\nabla_{\x} \lpos \alpha_{\imath} \rpos_{\imath}
&=    \lpos  \nabla_{\x} \alpha_{\imath} \rpos_{\imath} - \frac{1}{V} \oiint\limits_{\partial \Omega_{\imath}(t)} \alpha_{\imath} \, d^2 \x  
\\
\nabla_{\x} \cdot \lpos \valpha_{\imath} \rpos_{\imath}
&=    \lpos  \nabla_{\x} \cdot \valpha_{\imath} \rpos_{\imath} - \frac{1}{V} \oiint\limits_{\partial \Omega_{\imath}(t)} \valpha_{\imath} \cdot d^2 \x  
\end{split}
\label{eq:spatial_av_3D}
\end{gather}
where $\alpha_{\imath}$ is a scalar density and $\valpha_{\imath}$ is a vector or tensor density, in each case within the $\imath$th phase, $V$ is a fixed volume, $\lpos \alpha_{\imath} \rpos_{\imath} = V^{-1} \iiint\nolimits_{\Omega_{\imath}(t)} \alpha_{\imath} dV$ 
is the volumetric phase average of $\alpha_{\imath}$, $\Omega_{\imath} (t)$ is the component of the volume occupied by the $\imath$th phase, and $\partial \Omega_{\imath} (t)$ is the surface or interface of the $\imath$th phase.   Eqs.\ \eqref{eq:spatial_av_3D} adopt the convention that the unit normal $\n$ points out of the $\imath$th phase, giving the negative sign. Originally formulated by analogy with the Reynolds transport theorem, 
direct proofs of \eqref{eq:spatial_av_3D} have been provided by differential calculus methods \cite{Whitaker_1967b, Whitaker_1969, Slattery_1972, Bachmat_1972, Cushman_1982, Cushman_1983, Whitaker_1985, Whitaker_1999, Slattery_1999, Takatsu_2017}, a convolution method using a phase indicator or generalized function \cite{Marle_1967, Gray_Lee_1977, Gray_etal_1993, Chen_1994}, and by Lagrangian coordinate transformation \cite{Howes_Whitaker_1985, Slattery_1999} (analogous to that given in  \ref{sect:Apx_Re_temp}). For compressible fluids, an additional correction term is needed \cite{Cushman_1983, Takatsu_2017}. 
These insights have been used to derive volume-averaged versions of conservation equations for flow in porous media and multifluid systems \cite[e.g.,][]{Anderson_Jackson_1967, Whitaker_1967a, Whitaker_1967b, Slattery_1967, Marle_1967, Whitaker_1969, Drew_1971, Slattery_1972, Whitaker_1973, Whitaker_1985, Chen_1994, Whitaker_1999, Slattery_1999}, and have subsequently revealed a plethora of spatial and temporal averaging theorems applicable to such systems, enabling averaging over volumes, surfaces and curves, accounting for the effects of different scales \cite[e.g.,][]{Slattery_1967, Marle_1967, Drew_1971, Slattery_1972, Bachmat_1972, Whitaker_1973, Gray_Lee_1977, Cushman_1982, Cushman_1983, Cushman_1984, Whitaker_1986, Bachmat_Bear_1986, Plumb_Whitaker_1988, Bear_Bachmat_1991, Ochoa-Tapia_etal_1993, Gray_etal_1993, Chen_1994, Quintard_Whitaker_1994a, Grau_Cantero_1994, He_Sykes_1996, Whitaker_1999, Slattery_1999, Davit_etal_2010, Gray_Miller_2013, Wood_2013, Pokrajac_deLemos_2015, Takatsu_2017}. 

While the spatial averaging theorem is not examined further here, its formulation confirms the existence of alternative mathematical formulations of the Reynolds transport theorem \eqref{eq:Re_tr}. These are explored in greater generality in the following sections.

\subsection{\label{sect:spat_Re_tr} Velocimetric-Spatial Reynolds Transport Theorem}

We now examine a different class of time-independent (steady) flow systems, to give the following theorem.
\begin{theorem}
Consider a time-independent (steady) continuum represented by Eulerian volumetric and velocimetric (phase space) coordinates, in which each local property of a fluid is specified as a function of the Cartesian velocity $\u = [u,v,w]^\top$ and position $\x = [x,y,z]^\top$. Let $\varphi(\u,\x)$ be the density of a conserved quantity within a fluid, expressed per unit of velocity and volume space. 
Let $\varphi(\u, \x)$ be continuous and continuously differentiable with respect to velocity and position throughout a defined position-dependent region of velocity space (the ``domain'') $\D(\x)$, for all velocities up to its boundary and all positions considered.
The total differential of the integral of $\varphi(\u, \x)$ over the domain, {relative to an enclosed, position-dependent, smoothly deforming region of velocity space (the ``velocimetric control volume'')} satisfies:
\begin{gather}
\begin{split}
d \iiint\limits_{\D(\x)} \varphi \, d^3 \u 
&
=   
\biggl[
\iiint\limits_{\D(\x )}
\nabla_{\x} \varphi \; d^3 \u
+
\oiint\limits_{\partial \D(\x)} \varphi \, \vgrad_{rel}^\top \cdot d^2 \u
\biggr] \cdot d \x
\\
&=    
\biggl[
\iiint\limits_{\D(\x )}
\bigl(
\nabla_{\x} \varphi
+
\nabla_{\u} \cdot ( \varphi \, \vgrad_{rel}^\top )
\bigr) d^3 \u
\biggr] \cdot
d \x,
\end{split}
\label{eq:Re_tr_u_3D}
\end{gather}
where 
$\partial \D(\x)$ is the domain boundary (velocity surface), 
$d$ is the differential operator,
$d^3 \u = dudvdw$ is an infinitesimal velocimetric element within the domain, 
$d^2 \u = \n_B dB$ is an infinitesimal directed area element on the velocimetric boundary, 
in which $dB$ is the boundary area element and $\n_B$ is its outward unit normal,
$d \x = [dx,dy,dz]^\top$ is the differential of (vector) position, 
and
$\vgrad_{rel}=\vgrad_{rel}(\u,\x):= \nabla_{\x} \u_{rel}$ is the relative velocity gradient tensor field\footnote{In fluid mechanics, the velocity gradient tensor field is commonly denoted $\protect{{\partial \u}/{\partial \x}}$ or $\protect{\nabla_{\x} \u}$ with components $\protect{{\partial u_j}/{\partial x_i}}$. However, such notation creates  confusion over its functional dependencies. To avoid this, we use a distinct symbol for the velocity gradient; for a longer discussion see  \ref{sect:Apx_Re_spatial}.}.
We here use the $\partial (\to)/\partial (\downarrow)$ convention for vector derivatives, 
hence $\nabla_{\x} = {\partial}/{\partial \x}$ is the spatial gradient operator, 
and 
$\nabla_{\u} = {\partial}/{\partial \u}$ is the gradient operator in velocity space, assuming Cartesian coordinates $\x$ and $\u$.
In \eqref{eq:Re_tr_u_3D}, the quantity $\vgrad_{rel}^\top \cdot d^2 \u = \vgrad_{rel} \, d^2 \u=\vgrad_{rel} \, \n_B dB$ contains the tensor-vector product, while for consistency with the derivative convention, the divergence operation is defined by $\nabla_{\u} \cdot ( \varphi \, \vgrad_{rel}^\top )= [\nabla_{\u}^\top (\varphi \, \vgrad_{rel}^\top )]^\top$.
\end{theorem}
\begin{proof}
Two distinct proofs of \eqref{eq:Re_tr_u_3D} are given in  \ref{sect:Apx_Re_spatial}. 
Eq.\ \eqref{eq:Re_tr_u_3D} can also be derived from the general exterior calculus formulation presented in \S\ref{sect:general}, and shown to reduce to a generalized vector calculus formulation. 
\end{proof}

We can describe \eqref{eq:Re_tr_u_3D} as a {\it three-dimensional velocimetric-spatial Reynolds transport theorem}, or more precisely a transformation theorem. 
Its formulation bears many similarities to analyses of molecular systems in 
phase space \cite{Pottier_2010} 
(examined in \S\ref{sect:ex_sys}), 
but the integration extends only over the velocity space, in sympathy with the common probabilistic description of turbulent flow \cite[e.g.,][]{Monin_Yaglom_1971a, Batchelor_1967, Pope_2000}. 
The quantity $\varphi(\u,\x)$ can be interpreted physically as the conserved quantity carried per unit of velocity and volume space by a fluid element with a velocity between $\u$ and $\u + d\u$ at the position between $\x$ and $\x + d\x$, i.e., the density in six-dimensional Eulerian phase space. Integration of $\varphi(\u,\x)$ over the velocity domain therefore gives the volumetric density $\alpha(\x)$ of the conserved quantity at this position. 

The physical interpretation of \eqref{eq:Re_tr_u_3D} is analogous to that for the temporal formulation \eqref{eq:Re_tr}: a differential change in the integral of a local quantity $\varphi(\u,\x)$ over the velocity space  
can be subdivided into changes which occur within the control volume coincident with the velocity domain $\D(\x)$, and changes which take place due to (spatial) translations into or out of the domain through the velocity surface $\partial \D(\x)$. Using a velocimetric form of the divergence theorem, this is equivalent to the sum of changes within the domain and changes arising from a velocity divergence term. 
The combined velocimetric integral in \eqref{eq:Re_tr_u_3D} thus assumes a continuous and continuously differentiable tensor field $\vgrad_{rel}$, which extends over the entire velocity space within which the domain $\D(\x)$ is embedded.  

In sympathy with the temporal formulation, \eqref{eq:Re_tr_u_3D} adopts a relative velocity gradient, which can be decomposed (assuming Cartesian velocity and position coordinates) into two components:
\begin{align}
\vgrad_{rel}
= \vgrad-  \vgrad_{CV}, 
\label{eq:vel_tensor_relative} 
\end{align}
where $\vgrad$ is the intrinsic field 
and $ \vgrad_{CV} $ is the spatial rate of change of the local velocity coordinate system, as represented by a velocimetric control volume. For consistency, this control volume must provide a smoothly-varying tensorial frame of reference.
For flow of a  compressible Newtonian fluid, the intrinsic velocity gradient is related (implicitly) to the shear stress tensor field, here defined positive in compression \cite{Bird_etal_2002, White_2006}:
\begin{align}
\tens{\tau} =  - \mu (\vgrad  + \vgrad ^\top) - \lambda \, \vec{\delta} \, \Delta,
\label{eq:Newtonian}
\end{align}
where $\mu$ is the dynamic viscosity, $\lambda$ is the second or dilatational viscosity, $\vec{\delta}$ is the Kronecker delta tensor and $\Delta = \nabla_{\x} \cdot \u$ is the divergence of the velocity field.

In consequence, for this category of flow systems expressed using $(\u,\x)$ coordinates, the velocity gradient tensor  field $\vgrad$ -- or equivalently, the shear stress tensor field $\tens{\tau}$ -- provides an intrinsic spatial connection between different velocimetric domains. This is similar to the way in which, for a flow system described by $(\x,t)$ coordinates, the velocity field provides an intrinsic temporal connection -- a transport equation -- between different volumetric domains.

We have not been able to identify any previous report of the velocimetric-spatial Reynolds transport theorem \eqref{eq:Re_tr_u_3D} in the scientific literature.

\subsection{\label{sect:spat_prob} Probabilistic Analysis and the Spatial Liouville Equation}

Now consider a probabilistic form of the spatial formulation, based on the three-dimensional random variable for the velocity vector $\vect{\Upsilon}_{\u} = [\Upsilon_u, \Upsilon_v, \Upsilon_w]^\top$ with values $\u$,  subject to the three-dimensional random variable for position $\vect{\Upsilon}_{\x} = [\Upsilon_x, \Upsilon_y, \Upsilon_z]^\top$ with values $\x$.  
{We then define the joint-conditional pdf $f(\u|\x)$ over the domain $\D(\x)$, to represent the probability that the fluid element infinitesimally close to the specified position $\x$ will have a velocity infinitesimally close to $\u$:}
\begin{equation}
\begin{split}
f (\u|\x) \, d^3 \u = f (u,v,w|x,y,z) \, dudvdw 
\approx \text{Prob}
&\Biggl(
\begin{matrix} 
u \le \Upsilon_u \le u+du \\
v \le \Upsilon_v \le v+dv \\
w \le \Upsilon_w \le w+dw 
\end{matrix} 
\Biggl | \;
\begin{matrix} 
x \le \Upsilon_x \le x+dx \\
y \le \Upsilon_y \le y+dy \\
z \le \Upsilon_z \le z+dz
\end{matrix} 
\Biggr).
\end{split}
\label{eq:pdf_uvect_sub_xvect}
\end{equation}
{(Formally, $f(\u|\x)$ is defined over continuous intervals of velocity and position, from which \eqref{eq:pdf_uvect_sub_xvect} is obtained in the infinitesimal limits \cite[e.g.,][]{Feller_1966}.)}
Although not usually written in conditional form, we recognize $f (\u|\x)$ -- more commonly written $f (\u|\vec{r})$ as a function of relative position $\vec{r}$ -- as the basis of the Reynolds-averaged Navier-Stokes formulation, and the single-position correlation functions of turbulent fluid mechanics \cite[e.g.,][]{Monin_Yaglom_1971a, Batchelor_1967, Pope_2000, Hinze_1975}.  

Taking the velocity domain $\D(\x) \subseteq \mathbb{R}^3$ to be a function of $\x$, the pdf will be normalized at each position $\x$:
\begin{align}
1 = \iiint\limits_{\D(\x)} \, f (\u|\x) \, d^3 \u.
\label{eq:local_norm}
\end{align}
For any observable $\obss(\u, \x)$ in phase space coordinates, we can also define the conditional expectation:
\begin{align}
\lvel {\obss} \rvel (\x) = \iiint\limits_{\D(\x)} \obss (\u, \x) \, f (\u|\x) \, d^3 \u . 
\label{eq:velocity_average}
\end{align}
This can be interpreted physically as the ensemble mean of $\obss(\u,\x)$, i.e.\ its average over all values of the instantaneous velocity $\u \in \D(\x)$ at position $\x$, enabling a statistical rather than strict definition at each point. In many studies, \eqref{eq:velocity_average} is assumed equivalent to the local time mean $\overline{\obss} (\x)$. In the present work, we maintain the most general interpretation of \eqref{eq:velocity_average}, without any ergodic hypothesis. 

Now if $f(\u |\x)$ satisfies the same regularity conditions as the density $\varphi(\u,\x)$ in \S\ref{sect:spat_anal}, we can substitute $f(\u |\x)= \varphi(\u,\x)$ into the spatial Reynolds transport theorem \eqref{eq:Re_tr_u_3D}, giving:
\begin{align}
\begin{split}
d \iiint\limits_{\D(\x)} f \, d^3 \u 
=    
\biggl[
\iiint\limits_{\D(\x )}
\bigl(
\nabla_{\x} f
+
\nabla_{\u} \cdot ( f \, \vgrad_{rel}^\top )
\bigr) d^3 \u
\biggr] \cdot
d \x,
\end{split}
\label{eq:Re_tr_u_3D_prob}
\end{align}
From \eqref{eq:local_norm}, the left-hand side of \eqref{eq:Re_tr_u_3D_prob} vanishes for all $\x$. 
This gives the following theorem:
\begin{theorem}
Let $f(\u| \x)$ be the probability density of the velocity $\u$ at the specified position $\x$, defined over a velocity domain $\D(\x)$ containing a relative velocity gradient tensor field $\vgrad_{rel}$. Let $f(\u| \x)$ be continuous and continuously differentiable with respect to velocity and position throughout $\D(\x)$, for all velocities up to its boundary and all positions considered. For all $\u \in \D(\x)$ and all positions $\x$:
\begin{align}
\nabla_{\x} f  +  \nabla_{\u} \, \cdot ( f \, \vgrad_{rel}^\top ) 
= \vect{0}.
\label{eq:Liouville_3D}
\end{align}
\end{theorem}
\begin{proof}
Eq.\ \eqref{eq:Liouville_3D} follows directly from  \eqref{eq:local_norm}, \eqref{eq:Re_tr_u_3D_prob} and the fundamental lemma of the calculus of variations \cite{Weinstock_1952, Gelfand_Fomin_1963}, for all choices of compactly supported continuous and continuously differentiable functions $f(\u| \x)$.  
\end{proof}

\begin{corollary}
From \eqref{eq:Liouville_3D}, each spatial component must independently vanish:
\begin{align}
\frac{\partial f}{\partial x_i}  +  \nabla_{\u} \, \cdot \bigl(f  \vgrad_{rel,i}^\top \bigr)  
= 0,
\qquad
\forall x_i \in [x,y,z],
\label{eq:Liouville_1D_each}
\end{align}
where $ \vgrad_{rel,i} = \partial \u/\partial x_i$ is the $i$th row of the velocity gradient tensor. 
\end{corollary}

Eq.\ \eqref{eq:Liouville_3D} can be interpreted as a {\it three-dimensional spatial Liouville equation}. 
We again emphasize that if the pdf is not continuous or continuously differentiable, e.g., due to a discontinuity in the velocity gradient, \eqref{eq:Liouville_3D} may be invalid.
%

Despite an extensive search, we have not identified any previous report of the three-dimensional or one-dimensional spatial Liouville equation \eqref{eq:Liouville_3D}-\eqref{eq:Liouville_1D_each} in the fluid mechanics or physics literature, or even in the probability literature. A contributing factor may be that in the traditional Liouville equation derived by Gibbs \cite{Gibbs_1884}, based on the pdf $f(\pos,\mom|t)$ in $6N$-dimensional phase space (where $\pos$ and $\mom$ are position and momentum vectors), all parameters are functions of time, leading exclusively to a temporal Liouville equation
\cite[e.g.,][]{Tolman_1938, Landau_Lifshitz_1969} (see also \S\ref{sect:ex_sys}).
Spatial Liouville equations are also accessible using the apparatus of exterior calculus (see \S\ref{sect:general}), but we have not found any previous study to do so -- noting that this requires a multiparameter Lie derivative, divorcing this operator from the concept of physical time.

\subsection{\label{sect:spat_ops} Spatial Perron-Frobenius and Koopman Operators}

Taking the analysis farther, the solution to \eqref{eq:Liouville_3D} is $f(\u|\x)=\hat{{P}}_{\x} \, f(\u| \vec{0} )$, using a three-dimensional spatial Perron-Frobenius operator $\hat{{P}}_{\x}$, in which the origin $\x=\0$ is defined in the relative coordinate system of $\x$. Again this is linear, giving $\hat{{P}}_{\x} = \protect{\exp(\x \cdot  \, \hat{\vec{L}}_{\x})}$, in which $\hat{\vec{L}}_{\x}$ is a vector spatial Liouville operator defined by $\hat{\vec{L}}_{\x} f = - \nabla_{\u} \, \cdot ( f \, \vgrad_{rel}^\top )$.
The adjoint three-dimensional spatial Koopman operator $\hat{{K}}_{\x}$ can be defined from the ensemble average \eqref{eq:velocity_average} using the duality:
\begin{align}
\begin{split}
\lvel{\obss} \rvel (\x)
&= \iiint\limits_{\D(\x)} \obss(\u, \x) \, \hat{P}_{\x} \, f(\u|\vec{0}) \, d^3 \u 
\wrap
= \iiint\limits_{\D(\x)} \hat{K}_{\x} \, \obss(\u,\vec{0}) \, f (\u|\x) \, d^3 \u,
\end{split}
\label{eq:op_duality_spat}
\end{align}
giving the spatial evolution equation $\obss(\u,{\x})=\hat{{K}}_{\x} \, \obss(\u,\vec{0})$.  As noted, spatial Koopman operators or related methods have been invoked by several authors for the analysis of spatial symmetries in fluid flow systems \citep[e.g.,][]{Sharma_etal_2016, Hemati_etal_2017, Arbabi_Mezic_2017, Proctor_etal_2018, Kutz_etal_2018, LeClainche_Vega_2018a, LeClainche_Vega_2018b, Giannakis_etal_2019, Froyland_etal_2020, Umbria_etal_2020, Giannakis_Das_2020, Perez_etal_2020, Bai_etal_2020}.

In consequence, if one has local information on a time-independent flow system at one position, either in probabilistic form or in the form of a conserved observable property, it is possible to extrapolate this information using the spatial Perron-Frobenius and Koopman operators to all positions within the velocity gradient field. As shown in \S\ref{sect:ex_sys}, these operators can be further extended to spatiotemporal systems.

\subsection{\label{sect:spat_other} Further Simplifications}

{\bf Intrinsic Velocity Gradients}: For a fixed velocity gradient frame of reference $\vgrad_{CV} =0$, we obtain the intrinsic spatial Liouville equation:
\begin{align}
\nabla_{\x} f  +  \nabla_{\u} \, \cdot (f \vgrad^\top)  
= \vect{0},
\label{eq:Liouville_3D_intr}
\end{align}
expressing the natural variation of $f$ with $\x$. Eq.\ \eqref{eq:Liouville_3D_intr} can be considered to be induced by
$\vgrad^\top = \tens{\Xi} (\u)$, 
a system of spatial partial differential equations with tensor operator $\tens{\Xi}$.  

For incompressible or solenoidal flow $\nabla_{\x} \cdot \u = 0$ of a Newtonian fluid with a symmetric shear stress tensor, reduction of \eqref{eq:Newtonian} to the explicit relation $\vgrad = -  \tens{\tau}/2 \mu  $ gives the simplified spatial Liouville equation:
\begin{align}
\nabla_{\x} f  -  \nabla_{\u} \, \cdot f \frac{   \tens{\tau}^\top}{2 \mu}  
= \vect{0},
\label{eq:Liouville_3D_intr_Newt}
\end{align}
or if the shear stress tensor is homogenous in velocity space:
\begin{align}
\nabla_{\x} f  -  \frac{   \tens{\tau}^\top}{2 \mu}   \cdot  \nabla_{\u} f 
= \vect{0}.
\label{eq:Liouville_3D_intr_Newt2}
\end{align}

\noindent {\bf One-Dimensional Geometries}: For flows with a one-directional velocity gradient aligned with one $x_i$ from $[x_1,x_2,x_3]$, for example plane parallel flow in the zone of established flow, the above analysis will reduce to a one-dimensional spatial Reynolds theorem, which can be written as the total derivative:
\begin{align}
\begin{split}
\frac{d}{dx_i} \iiint\limits_{\D(\x)} \varphi \, d^3 \u 
&=   
\iiint\limits_{\D(\x )}
\frac {\partial \varphi}{\partial x_i} \; d^3 \u
+
\oiint\limits_{\partial \D(\x)} \varphi \, \vgrad_{rel,i}^\top  \cdot  d^2 \u
\wrap
=    
\iiint\limits_{\D(\x )}
\biggl[
\frac {\partial \varphi}{\partial x_i}
+
\nabla_{\u} \cdot ( \varphi \, \vgrad_{rel,i}^\top ) \biggr] d^3 \u.
\end{split}
\label{eq:Re_tr_u_1D}
\end{align}
Substituting $\varphi=f$, we obtain a single one-dimensional spatial Liouville equation in the form of \eqref{eq:Liouville_1D_each}. 

\vspace{10pt}

\noindent {\bf Two-Dimensional Velocity Gradients}: A special case involves a two-dimensional velocity vector $\u=[u,v]^\top $ with gradients in a single direction $x_i$, leading to a velocity gradient vector $(d \u/d x_i)_{rel} = [(d u/d x_i)_{rel},(d v/d x_i)_{rel}]^\top$ (here reverting to traditional notation). If this is also subject to the no-divergence condition $\protect{\nabla_{\u} \cdot (d \u/ d x_i)_{rel} =0}$, we can define a spatial stream function or {\it velocity gradient function} $\Gamma$ by 
$(d \u/dx_i)_{rel} = [{\partial \Gamma}/{\partial v},  - {\partial \Gamma}/{\partial u}]^\top$.
Substitution into the one-dimensional spatial Liouville equation \eqref{eq:Liouville_1D_each} gives the Hamiltonian form:
\begin{align}
\frac{\partial f  }{\partial x_i}   +\biggl( \frac{\partial \Gamma}{\partial v} \frac{\partial f}{\partial u}  -  \frac{\partial \Gamma}{\partial u} \frac{\partial f}{\partial v} \biggr)
     =0.
\label{eq:Liouville_new2D}
\end{align}
For a curl-free gradient $\nabla_{\u} \times (d\u/dx_i)_{rel} =\vec{0}$, the function $\Gamma$ will be normal to a {\it velocity gradient potential} $\Lambda$ defined by $(d\u/dx_i)_{rel} = \nabla_{\u} \Lambda=[{\partial \Lambda}/{\partial u}, {\partial \Lambda}/{\partial v}]^\top$. 
For constant gradients these give a {\it gradient net} of curvilinear orthogonal coordinates $(\Gamma,\Lambda)$, tangential and normal to the velocity gradient vector.

\section{General Formulations} \label{sect:general} 
\subsection{Exterior Calculus Formulations} \label{sect:diffform} 
{\bf Reynolds Transport Theorem on a Manifold:}
The above analyses are now generalized to the analysis of differential forms, using an extended multivariate formulation of exterior calculus.

\begin{theorem} 
\label{thm:ReTT_ext_calc}
Consider an $n$-dimensional orientable differentiable manifold $M^n$, described using a patchwork of local coordinate systems, in which there is an $r$-dimensional oriented compact submanifold $\Omega^r$. Let $\V$ be a smooth vector or tensor field in $M^n$, parameterized by an $m$-dimensional parameter vector $\C$, but not itself a function of $\C$. 
Let $\omega^r$ represent a field of $r$-forms in $M^n$ associated with a conserved quantity, which is locally continuous and continuously differentiable with respect to the local coordinates and parameter $\C$ within $\Omega^r$, for all coordinates up to its boundary and all parameter values considered.
The integral of $\omega^r$ over the submanifold satisfies:
\begin{gather}
\begin{split}
d \int\limits_{\Omega(\C)}   \omega^r 
&=  \biggl[ \int\limits_{\Omega(\C)} \mathcal{L}_{\V}^{(\C)}  \omega^r \biggr] \cdot {d\C} 
\\&
=  \biggl[ \int\limits_{\Omega(\C)}  \, i_{\V}^{(\C)}  \, d  \omega^r + {\oint\limits_{\partial \Omega(\C)}  }  \,   i_{\V}^{(\C)}  \, \omega^r 
\biggr]  \cdot {d\C} 
\wrap
=   \biggl[  \int\limits_{\Omega(\C)}  
 i_{\V}^{(\C)}  \, d \omega^r + d  ( i_{\V}^{(\C)} \, \omega^r  )
\biggr] \cdot  {d\C},
\end{split}
\label{eq:Re_tr_diffform}
\end{gather}
where $d$ is now the exterior derivative, $\partial \Omega(\C)$ is the submanifold boundary, $\mathcal{L}_{\V}^{(\C)}$ is a multiparameter Lie derivative with respect to $\V$ over parameters $\C$, $i_{\V}^{(\C)}$ is a multiparameter interior product with respect to $\V$ over parameters $\C$, and 
``$\cdot$'' is the usual vector scalar product (dot product). 
\end{theorem}
\begin{proof}
The definitions of the multiparameter operators used in \eqref{eq:Re_tr_diffform}, and its proof, are provided in  \ref{sect:Apx_Re_exterior}. 
\end{proof}

Eq.\ \eqref{eq:Re_tr_diffform} can be recognized as a {\it generalized parametric Reynolds transport theorem} -- or more precisely, a transformation theorem -- applicable to a differential form associated with a conserved quantity on a manifold.  For flow in a geometric space with $\C=t$, \eqref{eq:Re_tr_diffform} reduces to the one-parameter exterior calculus formulation of the Reynolds transport theorem, based on the time-independent velocity field $\u$  [\onlinecite[][\S 9.2]{Lee_2009}; \onlinecite[][eqs.\ 0.49, 4.33-4.34]{Frankel_2013}].

{\bf Augmented Reynolds Transport Theorem on a Manifold:} In many flow systems, the vector or tensor field $\V$ will also be a function of the parameter $\C$. This can be handled by augmenting the manifold $M^n$ with the parameter space $\R^m$, giving the augmented field represented by $ \V \comp \C$, where $\comp$ here denotes an augmentation operator [c.f. \onlinecite{Flanders_1973, Frankel_2013}]. Eq.\ \eqref{eq:Re_tr_diffform} then applies based on the field $\V \comp \C$ in the augmented manifold $M^n \times \R^m$. This yields the following theorem:

\begin{theorem}  
Consider an $n$-dimensional orientable differentiable manifold $M^n$ containing an $r$-dimensional oriented compact submanifold $\Omega^r$ and a smooth vector or tensor field $\V$, all defined as in Theorem \ref{thm:ReTT_ext_calc}. Let the vector or tensor field now be a function of $\C$. Let $\omega^r$ be a field of $r$-forms in $M^n$, defined as in Theorem \ref{thm:ReTT_ext_calc}. 
The integral $\omega^r$ over the submanifold satisfies:
\begin{align}
\begin{split}
\hat{d} \int\limits_{\Omega(\C)}  \omega^r 
=  \biggl[ \int\limits_{\Omega(\C)} \mathcal{L}_{\V \comp \C}^{(\C)}  \, \omega^r \biggr] \cdot {d\C} 
&=  \biggl[ \int\limits_{\Omega(\C)}  \vpartial_{\C} \omega^r 
+ \int\limits_{\Omega(\C)}  \, i_{\V}^{(\C)}  \, d  \omega^r 
+ {\oint\limits_{\partial \Omega(\C)}  }  \,   i_{\V}^{(\C)}  \, \omega^r 
\biggr]  \cdot {d\C} 
\\&=   \biggl[  \int\limits_{\Omega(\C)}  
\vpartial_{\C} \omega^r +  i_{\V}^{(\C)}  \, {d} \omega^r + {d} ( i_{\V}^{(\C)} \, \omega^r  )
\biggr] \cdot  {d\C} 
\end{split}
\label{eq:Re_tr_diffform_aug}
\end{align}
where $\hat{d}$ is an extended exterior derivative based on the augmented coordinates, and $\vpartial_{\C} = [\partial/\partial C_1, ..., \partial/\partial C_m]^\top$ is a vector partial differential operator with respect to $\C$. For consistency, \eqref{eq:Re_tr_diffform_aug} retains the notation ${d}$ for the exterior derivative based on the standard local coordinates. 
\end{theorem}
\begin{proof}
The proof of \eqref{eq:Re_tr_diffform_aug} is given in  \ref{sect:Apx_Re_exterior_ext}.  
\end{proof}

Eq.\ \eqref{eq:Re_tr_diffform_aug} can be described as an {\it augmented generalized parametric Reynolds transport theorem}, based on the vector or tensor field $\V(\C)$. When $\C$ is expressed in Cartesian coordinates, the first term in the integrand of \eqref{eq:Re_tr_diffform_aug2} can be written as $\nabla_{\C}  \omega^r$. For flow in a geometric space with $\C=t$, \eqref{eq:Re_tr_diffform_aug} reduces to the augmented one-parameter exterior calculus formulation of the Reynolds transport theorem, based on the time-varying velocity field $\u(t)$ [\onlinecite{Flanders_1973}, eq.\ 7.2, \onlinecite{Frankel_2013}, eqs.\ 0.50, 4.42].

\vspace{10pt}
{\bf Liouville Theorem and Operators on a Manifold:}
Extending the analyses presented in \S\ref{sect:temp_anal} and \S\ref{sect:spat_anal}, the above Reynolds transport theorems can be used to derive probabilistic differential equations and corresponding operators. Consider a field of probability $r$-forms represented by $\rho^r = \rho^r_{\C}$, a function of a patchwork of local coordinate systems over the submanifold $\Omega(\C)$, and conditioned on the parameter $\C$. This can be defined by the following axioms:
\begin{align}
\rho^r \ge 0
\hspace{10pt} 
\text{ and } 
\hspace{5pt}
\int\limits_{\Omega(\C)} \rho^r =1
\label{eq:prob_form_constr}
\end{align}
The underlying subtleties in this definition -- arising from its marriage of measure theory and exterior calculus -- are discussed in  \ref{sect:Apx_prob_forms}. 
We further define the expected value of a scalar field (0-form) $\obsss$ over the submanifold by: 
\begin{align}
\lgen {\obsss} \rgen (\C)  = \int\limits_{\Omega(\C)} \obsss \, \rho^r
\label{eq:submanifold_average}
\end{align}
The Reynolds transport theorem on a manifold \eqref{eq:Re_tr_diffform} then leads to the following theorem:

\begin{theorem} 
\label{thm:Liouville_eq_ext_calc}
Let $\rho^r$ be a field of $r$-forms, representing the probability density at each position in an $n$-dimensional orientable differentiable manifold $M^n$, as a function of the local coordinates and conditioned on the $m$-dimensional parameter vector $\C$. 
Let $\rho^r$ be locally continuous and continuously differentiable with respect to the local coordinates and parameter $\C$ within an $r$-dimensional oriented compact submanifold $\Omega(\C)$ in $M^n$, for all coordinates up to its boundary and all parameter values considered.
Let $\V$ be a vector or tensor field in $\Omega(\C)$, which is independent of $\C$. 
For each point in $\Omega(\C)$ and all $\C$:
\begin{gather}
 \mathcal{L}_{\V}^{(\C)} \rho^r
=  i_{\V}^{(\C)}  \, d \rho^r + d \,  (i_{\V}^{(\C)} \, \rho^r )
= \vec{0}.
 \label{eq:Liouville_diffform}
\end{gather}
\end{theorem}

\begin{proof}
Substituting $\omega^r= \rho^r $ in \eqref{eq:Re_tr_diffform}, the left-hand side gives $d \int\nolimits_{\Omega(\C)} \rho^r$, the exterior derivative of a $0$-form, equivalent to its differential. By normalization \eqref{eq:prob_form_constr}, this vanishes for all $\C$. Eq.\ \eqref{eq:Liouville_diffform} follows from \eqref{eq:Re_tr_diffform} and the fundamental lemma of the calculus of variations (in an exterior calculus formulation), for all choices of compactly supported continuous and continuously differentiable probability forms $\rho^r$. 
\end{proof}

Eq.\ \eqref{eq:Liouville_diffform} can be interpreted as a {\it generalized parametric Liouville equation}, which expresses the local conservation of the probability $r$-form under the (intrinsic) variation of its conditions. We also recognise \eqref{eq:Liouville_diffform} as a multiparameter extension of the Cartan relation of exterior calculus \citep[e.g.][]{Frankel_2013}, applied to a probability form.
Its solution can be written as $\rho^r_{\C}=\hat{P}_{\C} \, \rho^r_{\0}$, which defines an exterior Perron-Frobenius operator $\hat{P}_{\C}$. 
From the submanifold average \eqref{eq:submanifold_average}, this will have an adjoint exterior Koopman operator $\hat{{K}}_{\C}$ defined by the observable map $\obsss_{\C}=\hat{{K}}_{\C} \, \obsss_{\0}$.

{\bf Augmented Liouville Theorem and Operators on a Manifold:}
As discussed, in many systems the field $\V$ is also a function of $\C$. Using the augmented Reynolds transport theorem \eqref{eq:Re_tr_diffform_aug}, we can extract the theorem:

\begin{theorem} 
Let $\rho^r$ be a field of probability $r$-forms in an $n$-dimensional orientable differentiable manifold $M^n$,
defined as in Theorem \ref{thm:Liouville_eq_ext_calc}. 
Let $\V$ be a vector or tensor field in the $r$-dimensional submanifold $\Omega(\C)$ in $M^n$, a function of the local coordinates and also of $\C$. For each point in $\Omega(\C)$ and all $\C$:
\begin{gather}
 \mathcal{L}_{\V \comp \C}^{(\C)} \rho^r
= \vpartial_{\C} \rho^r +  i_{\V}^{(\C)}  \, {d} \rho^r + {d} ( i_{\V}^{(\C)} \, \rho^r)
= \vec{0}.
 \label{eq:Liouville_diffform_aug}
\end{gather}
\end{theorem}
\begin{proof}
The proof of Theorem \ref{thm:Liouville_eq_ext_calc} applied to \eqref{eq:Re_tr_diffform_aug} gives \eqref{eq:Liouville_diffform_aug}, thus for all choices of compactly supported continuous and continuously differentiable probability forms $\rho^r$. 
\end{proof}

Eq.\ \eqref{eq:Liouville_diffform_aug} provides an augmented generalized parametric Liouville equation based on the probability $r$-form $\rho^r$, applicable for fields $\V(\C)$. Following the previous procedure, it can be used to define exterior Perron-Frobenius and Koopman operators for the augmented flow system.

Examining the literature, although multiparameter Lie groups and other generalizations have been examined \cite[e.g.,][]{Bluman_Kumei_1996, Baumann_2000, Oliveri_2010}, 
neither the full multiparameter Reynolds transport theorems \eqref{eq:Re_tr_diffform} or \eqref{eq:Re_tr_diffform_aug} nor their associated Liouville equations \eqref{eq:Liouville_diffform} or  \eqref{eq:Liouville_diffform_aug} appear to have been reported previously. 
As noted, one-parameter exterior calculus formulations of the Reynolds transport theorem have been reported \cite{Flanders_1973, Frankel_2013}. 
A temporal Liouville equation has also been written for an arbitrary conserved $r$-form in terms of the standard Lie derivative \cite[e.g.,][]{Nash_2015}, but not (we believe) in terms of a probability $r$-form. 

\subsection{Parametric Formulations} \label{sect:paramform} 
{\bf Parametric Reynolds Transport Theorem:} 
The augmented version of the exterior calculus formulation of the Reynolds transport theorem leads to the following theorem:
\begin{theorem}
Consider an $n$-dimensional space $M$ described using global Cartesian coordinates $\X$, containing an $n$-dimensional compact domain $\Omega$. 
Let $\V=(\nabla_{\C} \X)^\top$ be a smooth vector or tensor field in $M$, where $\C$ is an $m$-dimensional Cartesian parameter vector, such that $\V$ is a function of $\C$. 
Let $\omega^n$ be an $n$-dimensional compact material form based on the density $\psi(\X,\C)$ of a conserved quantity in $M$. Let $\psi(\X, \C)$ be continuous and continuously differentiable with respect to $\X$ and $\C$ throughout the domain $\Omega(\C)$, for all coordinates up to its boundary and all parameter values considered.
The integral of $\psi$ over the submanifold satisfies:
\begin{gather}
\begin{split}
d\int\limits_{\Omega(\C)}  \psi \, d^n \X 
&=  \biggl[  \int\limits_{\Omega(\C)}  \nabla_{\C} \psi  \; d^n \X +{\oint\limits_{\partial \Omega(\C)} }  \; \psi \, \V \cdot d^{n-1} \X  \biggr] \cdot {d\C} 
\\&
=    \biggl[ \int\limits_{\Omega(\C)}  \bigl[ \nabla_{\C} \psi +  \nabla_{\X} \cdot \bigl(  \psi \, \V  \bigr ) \bigr] d^n \X  \biggr] \cdot {d\C},
\end{split}
\label{eq:Re_tr_vect_calc_gen}
\end{gather}
where $d^n\X$ is an $n$-dimensional volume element in $\Omega(\C)$, $d^{n-1}\X$ is an $(n-1)$-dimensional directed area element on the boundary $\partial \Omega(\C)$, and in which the gradient and divergence operators are extended to their $n$- and $m$-dimensional variants.
\end{theorem}
\begin{proof}
The proof of \eqref{eq:Re_tr_vect_calc_gen} is given in  \ref{sect:Apx_Re_param}. 
\end{proof}

Eq.\ \eqref{eq:Re_tr_vect_calc_gen} provides a generalized parametric Reynolds transport theorem for a system with global Cartesian coordinates $\X$ and parameters $\C$. 
We emphasise that $\V$ in \eqref{eq:Re_tr_vect_calc_gen} is defined as the reverse of the $\partial (\to)/\partial (\downarrow)$ convention,  consistent with \eqref{eq:Re_tr_diffform} and \eqref{eq:Re_tr_diffform_aug}; furthermore the divergence is defined by 
$ \nabla_{\X} \cdot (  \psi \, \V  ) =  [\nabla_{\X}^\top ( \psi \, \V )]^\top$. 

\vspace{10pt}
{\bf Parametric Liouville Equation and Operators:}
The parametric Reynolds transport theorem \eqref{eq:Re_tr_vect_calc_gen} leads directly to the following theorem:

\begin{theorem}
Let $\q(\X|\C)$ be the probability density in the compact domain $\Omega(\C)$, a function of the $n$-dimensional global Cartesian coordinates $\X$ and conditioned on the $m$-dimen\-sional Cartesian parameter $\C$. 
Let $\q(\X | \C)$ be continuous and continuously differentiable with respect to $\X$ and $\C$ throughout  $\Omega(\C)$, for all coordinates up to its boundary and all parameter values considered.
Let $\V$ be a smooth vector or tensor field defined by $\V:=(\nabla_{\C} \X)^\top$, such that $\V$ is a function of $\C$. 
For all $\X \in \Omega(\C)$ and all $\C$: 
\begin{gather}
\nabla_{\C} \, \q +  \nabla_{\X} \cdot \bigl( \q \, \V  \bigr )
= \vec{0}.
\label{eq:Liouville_vect_calc_gen}
\end{gather}
\end{theorem}
\begin{proof}
Substituting $\psi=\q$ in \eqref{eq:Re_tr_vect_calc_gen}, the left-hand side $d \int\nolimits_{\Omega(\C)} \q \, d^n \X$ vanishes by the normalization of $\q$. Eq.\ \eqref{eq:Liouville_vect_calc_gen} follows from the fundamental lemma of the calculus of variations, for all choices of compactly supported continuous and continuously differentiable pdfs $\q$. 
\end{proof}

Eq.\ \eqref{eq:Liouville_vect_calc_gen} gives a parametric Liouville equation in a global coordinate system, expressing the conservation of the pdf $\q(\X | \C)$ under the (intrinsic) variation of its parameters $\C$. 
This can be considered to be induced by the dynamical system $\V = \vec{\mathcal{F}}(\X,\C)$ with operator $\vec{\mathcal{F}}$. Its solution can be written $\q(\X|\C)=\hat{P}_{\C} \, \q(\X|\vec{0})$ using a linear parametric Perron-Frobenius operator $\hat{P}_{\C} =  \exp(\C \cdot  \hat{\vec{L}}_{\C})$, in which the origin $\C=\0$ is measured in the relative coordinate system of $\C$, and $\hat{\vec{L}}_{\C}$ is a parametric Liouville operator defined by $\hat{\vec{L}}_{\C} \q= - \nabla_{\X} \cdot \bigl(  \q \, \V  \bigr ) $. The adjoint parametric Koopman operator $\hat{{K}}_{\C}$, defined using the moment $\lgenn {\varsigma} \rgenn (\C) = \int\nolimits_{\Omega(\C)} \varsigma \q \, d^n \X$, gives the observable equation $\varsigma(\X,{\C})=\hat{{K}}_{\C} \, \varsigma(\X,\vec{0})$.

\section{Applications}  \label{sect:ex_sys} 

We now have the apparatus to construct multidimensional parametric Reynolds transport theorems, 
Liouville equations and evolution operators for a variety of physical systems.  Several examples are examined (in intrinsic form) below. 


\subsection{Velocimetric spatiotemporal fluid flow systems} \label{sect:spatiotempvel} These apply to spatially and time-varying flows described by the density $\varphi(\u(\x,t),\x,t)$, giving the Reynolds transport equation:
\begin{align}
\begin{split}
&d \iiint\limits_{\D(\x,t)} \varphi \, d^3 \u 
\\
&=   \biggl[ \iiint\limits_{\D(\x,t)} \nabla_{\x} \varphi \; d^3 \u + \oiint\limits_{\partial \D(\x,t)} \varphi \, \vgrad^\top  \cdot  d^2 \u \biggr] \cdot d \x
+ \biggl[ \iiint\limits_{\D(\x,t)} \frac{\partial \varphi}{\partial t} \; d^3 \u + \oiint\limits_{\partial \D(\x,t)} \varphi \, \dot{\u}  \cdot  d^2 \u \biggr] dt
\\
&=   \biggl[ \iiint\limits_{\D(\x,t)} \bigl( \nabla_{\x} \varphi + \nabla_{\u} \cdot ( \varphi \, \vgrad ^\top) \bigr) d^3 \u \biggr] \cdot
d \x 
+  \biggl[ \iiint\limits_{\D(\x,t)} \biggl( \frac{\partial \varphi}{\partial t} + \nabla_{\u} \cdot ( \varphi \, \dot{\u} ) \biggr) d^3 \u \biggr] dt,
\end{split}
\label{eq:ReTT_uxt}
\end{align}
where $\D(\x,t)$ is the domain and $\dot{\u}= {\partial \u}/{\partial t}$. Introducing the pdf $\f(\u | \x, t)$, the corresponding joint Liouville equations are:
\begin{align}
&\begin{cases}
\nabla_{\x} \f   +  \nabla_{\u} \, \cdot ( \f \, \vgrad^\top )     = \vect{0} \\
\dfrac{\partial \f}{\partial t}   +  \nabla_{\u} \cdot (\f \, \dot{\u} )=0
\end{cases},
\label{eq:Liouville_uxt}
\end{align}
Introducing the four-dimensional operator $\vec{\delfour}_{\x} = [\partial/\partial x, \partial/\partial y, \partial/\partial z, \partial/\partial t]^\top$ and tensor-vector field $\vgradtil = \delfour_{\x} \u$, the latter can be written as:
\begin{align}
\delfour_{\x} \f   +  \nabla_{\u} \, \cdot ( \f \, \vgradtil^\top )     = \vect{0}.
\label{eq:Liouville_uxt_4D}
\end{align}
This can be considered induced by $\vgradtil^\top=\tens{\mathcal{F}}(\u)$. The Liouville equation \eqref{eq:Liouville_uxt_4D} and moment $\lvel {\obss} \rvel $ \eqref{eq:velocity_average} based on the observable $\obss(\u,\x,t)$ then give the spatiotemporal maps $\f(\u | \x,t)=\hat{{P}}_{\x,t} \, \f(\u | \0,0)$ and $\obss(\u,{\x},t)=\hat{{K}}_{\x,t} \, \obss(\u,\0,0)$, invoking the spatiotemporal operators
$\hat{P}_{\x,t} = \exp([\x,t] \,  \hat{\vec{L}}_{\x,t})$,
$\hat{\vec{L}}_{\x,t} \, \f= - \nabla_{\u} \cdot ( \f \, \vgradtil^\top) = - \nabla_{\u} \cdot ( \f \, \tens{\mathcal{F}}(\u)) $ and 
$\hat{{K}}_{\x,t}$.
The connections between these operators and those examined recently in the literature \cite[e.g.,][]{Sharma_etal_2016, Hemati_etal_2017, LeClainche_Vega_2018a, LeClainche_Vega_2018b} warrant further examination.


\subsection{Velocimetric spatiotemporal fluid flow systems with pairwise correlations} Moving directly to a probabilistic framework, these invoke the pairwise pdf 
$\f(\u_1(\x_1,t), $ $ \u_2(\x_2,t) $ $|\x_1, \x_2,t)$
\cite[e.g.,][]{Batchelor_1967, Monin_Yaglom_1971a, Pope_2000}, giving the Liouville equation system:
\begin{align}
\begin{cases}
\left. \begin{array}{llll}
\nabla_{\x_1} \f   
&+  \nabla_{\u_1} \, \cdot ( \f \, \vgrad_{\x_1}^\top )  
&    
&= \vect{0} 
\\
\nabla_{\x_2} \f   
&  
&+  \nabla_{\u_2} \, \cdot ( \f \, \vgrad_{\x_2}^\top )     
&= \vect{0} 
\\
\dfrac{\partial \f}{\partial t}   
&+  \nabla_{\u_1} \cdot ( \f \, \dot{\u}_1 )
&+  \nabla_{\u_2} \cdot ( \f \, \dot{\u}_2 )
&=0
\end{array} \right.
\end{cases},
\label{eq:Liouville_xtpair}
\end{align}
where $\nabla_{\x_k}$ is based on $\x_k$, $\dot{\u}_k= {\partial \u}_k/{\partial t}$ and $\vgrad_{\x_k}=\nabla_{\x_k} \u_k$. This expresses the dynamical system $\nabla_{\x_1,\x_2,t} [\u_1,\u_2] =\tens{\mathcal{F}}(\u_1,\u_2)$. Previous workers give only the temporal equation  \cite[e.g.,][]{Edwards_1964}. 
The Liouville equation \eqref{eq:Liouville_xtpair} and two-point moment $\lgenn {\obsss} \rgenn (\x_1,\x_2,t) $ based on the observable $\obsss(\u_1, \u_2, \x_1, \x_2,t)$ then give the pairwise maps $\f(\u_1,\u_2 | \x_1,\x_2,t)=\hat{{P}}_{\x_1,\x_2,t} \, \f(\u_1,\u_2 | \0,\0,t)$ and $\obsss(\u_1,\u_2, \x_1,\x_2,t)$ $=\hat{{K}}_{\x_1,\x_2,t} \, $ $\obsss(\u_1,\u_2, \0,\0,0)$, using the pairwise 
operators 
$\hat{P}_{\x_1,\x_2,t} = $ $\exp([\x_1,\x_2,t]  \hat{\vec{L}}_{\x_1,\x_2,t})$, 
$\hat{\vec{L}}_{\x_1,\x_2,t} \, \f$ $= - \nabla_{\u_1,\u_2} \cdot (  \f \, \tens{\mathcal{F}}(\u_1,\u_2) )$ and 
$\hat{{K}}_{\x_1,\x_2,t}$.

\hspace{10pt} For homogenous turbulence $\x_1 \mapsto \x, \x_2 \mapsto \x + \vec{r}, \u_1 \mapsto \u_0, \u_2 \mapsto \u_{\vec{r}}$, the pdf reduces to $\f(\u_0(t), \u_{\vec{r}}(\vec{r},t) | \vec{r},t)$ \cite{Batchelor_1967, Monin_Yaglom_1971a}, transforming \eqref{eq:Liouville_xtpair}:
\begin{equation}
\begin{cases}
\left. \begin{array}{llll}
\nabla_{\vec{r}} \f   
&  
& +  \nabla_{\u_{\vec{r}}} \, \cdot ( \f \, \vgrad_{\r}^\top )     
&= \vect{0} 
\\
\dfrac{\partial \f}{\partial t}   
&+  \nabla_{\u_0} \cdot (\f \, \dot{\u}_{\0} )
&+  \nabla_{\u_{\vec{r}}} \cdot (\f \, \dot{\u}_{\r} )
&=0
\end{array} \right.
\end{cases},
\label{eq:Liouville_rtpair}
\end{equation}
using $\vgrad_{\0}=\nabla_{\x} \u_{\0}=\0$ and $\vgrad_{\r}=\nabla_{\r} \u_{\r}$.
This is induced by the dynamical system $\nabla_{\r,t} [\u_{\0},\u_{\r}] $ $=\tens{\mathcal{F}}(\u_{\0},\u_{\r})$, and gives maps with 
simplified 
homogeneous operators
$\hat{P}_{\r,t} =  \exp([\r,t]  \hat{\vec{L}}_{\r,t})$, 
$\hat{\vec{L}}_{\r,t} \, \f = - \nabla_{\u_{\0},\u_{\r}} \cdot ( \f \, \tens{\mathcal{F}}(\u_{\0},\u_{\r})) $ and 
$\hat{{K}}_{\r,t}$.
In isotropic flow, the velocity gradient is constant in all directions $\vgrad_{\r}=\vec{\delta} \, d|\u_{\vec{r}}|/dr =- \tau \vec{\delta}/ \mu$, allowing further simplification.

\subsection{Velocimetric spatiotemporal fluid flow systems with $n$-wise correlations} The probabilistic framework for the previous system can be extended to triadic, quartic or $n$-wise correlations, based on the pdf $\f(\u_1(\x_1,t), ..., \u_n(\x_n,t) | \x_1, ..., \x_n,t)$ \cite{Batchelor_1967, Monin_Yaglom_1971a, Pope_2000}. The Liouville system is then:
\begin{equation}
\begin{cases}
\left. \begin{array}{llllll}
\nabla_{\x_k} \f   
&+  \nabla_{\u_k} \, \cdot ( \f \, \vgrad_{\x_k}^\top)  
&= \vect{0}, 
\hspace{5pt}
 \forall k \in \{1,...,n\}
\\
\dfrac{\partial \f}{\partial t}   
&+  \sum\limits_{k=1}^n \nabla_{\u_k} \cdot (\f \, \dot{\u}_k)
&=0
\end{array} \right.
\end{cases}.
\label{eq:Liouville_xtnwise}
\end{equation}
The dynamical system can be written as $\nabla_{\x_1,...,\x_n,t} [\u_1,...,\u_n] =\tens{\mathcal{F}}(\u_1,...,\u_n)$. The Liouville system \eqref{eq:Liouville_xtnwise} and $n$-point moments then give the maps $\f=\hat{{P}}_{\x_1,...,\x_n,t} \, \f_{\0}$ and $\obsss=\hat{{K}}_{\x_1,...,x_n,t} \, \obsss_{\0}$, 
based on the observable $\obsss(\u_1, ..., \u_n, \x_1, ..., \x_n,t)$ and the multipoint operators
$\hat{P}_{\x_1,...,\x_n,t} =  \exp([\x_1,...,\x_n,t]  \hat{\vec{L}}_{\x_1,...,\x_n,t})$, 
$\hat{\vec{L}}_{\x_1,...,\x_n,t}\, \f = $ 
$\protect{- \nabla_{\u_1,...,\u_n} \cdot (\f \, \tens{\mathcal{F}}(\u_1,...,\u_n))}$ and 
$\hat{{K}}_{\x_1,...,\x_n,t}$.

\subsection{Phase space systems (including molecular gases)} \label{sect:Hamilt} These can be described by the generalized phase space density $\varphi(\pos(t), \mom(t), t)$ and pdf $\rho(\pos (t) , \mom (t) |t)$ over the phase space $\Omega(t)$, based on the positions $\pos$ and momenta $\mom$ defined as $6N$-vectors to represent $N$ particles. These respectively give the phase space Reynolds transport theorem and Liouville equation (reverting to traditional notation):
\begin{align}
\begin{split}
\frac{d}{dt} \int\limits_{\Omega(t)} \varphi \, d^3 \pos \, d^3 \mom 
&= \int\limits_{\Omega(t)} \dfrac{\partial \varphi}{\partial t} \; d^3 \pos \, d^3 \mom
+ \oint\limits_{\partial \Omega(t)} \varphi \, \frac{d\pos}{dt} \cdot  d^2 \pos 
+ \oint\limits_{\partial \Omega(t)} \varphi \, \frac{d\mom}{dt}  \cdot  d^2 \mom 
\\
&=  \int\limits_{\Omega(t)} \biggl[ \dfrac{\partial \varphi}{\partial t} + \nabla_{\pos} \cdot \biggl( \varphi \, \frac{d\pos}{dt} \biggr) + \nabla_{\mom} \cdot \biggl( \varphi \, \frac{d\mom}{dt} \biggr) \biggr] d^3 \pos \, d^3 \mom 
\end{split}
\label{eq:ReTT_phase}
\\
&\frac{\partial \rho}{\partial t}   +  \nabla_{\pos} \cdot \biggl(\rho \, \frac{d \pos}{d t} \biggr) +  \nabla_{\mom} \cdot \biggl(\rho \, \frac{d \mom}{d t} \biggr)=0.
\label{eq:Liouville_phase}
\end{align}
The latter expresses the dynamical system $d [\pos,\mom] /d t= \tens{\mathcal{F}}(\pos,\mom)$. 
Indeed, the Boltzmann equation can be written in this form \cite{Harris_1971, Diu_etal_2001}. 
The Liouville equation \eqref{eq:Liouville_phase} and moment $\lgenn {\obsss}  \rgenn$ based on the observable $\obsss(\pos,\mom,t)$ then give the phase space maps $\rho(\pos,\mom | t)=\hat{{P}}_{t} \, \rho(\pos,\mom | 0)$ and $\obsss(\pos,\mom, t)=\hat{{K}}_{t} \, \obsss(\pos,\mom, 0)$, invoking temporal forms of the Perron-Frobenius and Koopman operators, 
now with Liouville operator 
$\hat{{L}}_{t}\, \rho= - \nabla_{\pos,\mom} \cdot (  \rho \, \tens{\mathcal{F}}(\pos,\mom)) $.

\hspace{10pt} Making the further assumption of zero-divergence 
flow \cite{Tuckerman_2003, Schuster_Just_2005} $\nabla_{\pos,\mom} \cdot [ \dot{\pos}, \ddot{\pos} ] =0$ (whence $\nabla_{\pos} \cdot \dot{\pos}=0$ and $\nabla_{\mom} \cdot \ddot{\pos} =0$) gives Liouville's theorem as written by Gibbs \cite{Gibbs_1884, Harris_1971}:
\begin{align}
\frac{d\rho}{dt} = \frac{\partial \rho}{\partial t}   + \biggl( \nabla_{\pos} \rho \cdot  \frac{d \pos}{d t}  +  \nabla_{\mom} \rho \cdot  \frac{d \mom}{d t} \biggr) =0.
\label{eq:Liouville_thm}
\end{align}
This is the oft-quoted statement of ``conservation of phase'', a special case of the more general result \eqref{eq:Liouville_phase}. Introducing the Hamiltonian $H(\pos,\mom)$ by the relations:
\begin{align}
\frac{d \pos}{d t} = \frac{\partial H}{\partial \mom}, \quad \frac{d \mom}{d t}  = - \frac{\partial H}{\partial \pos},
\label{eq:Hamilt}
\end{align}
this reduces to the Hamiltonian form \cite{Tolman_1938, Harris_1971, Pottier_2010}:
\begin{align}
\frac{d\rho}{dt} =\frac{\partial \rho}{\partial t}   + \biggl( \nabla_{\pos} \rho \cdot   \frac{\partial H}{\partial \mom}   - \nabla_{\mom} \rho \cdot  \frac{\partial H}{\partial \pos} \biggr)  = 0.
\label{eq:Liouville_Ham}
\end{align}
It is readily verified that the Hamiltonian form \eqref{eq:Hamilt} satisfies zero divergence \cite{Tuckerman_2003}.

\subsection{Lagrangian spatiotemporal fluid flow systems} These can described by the density $\wa(\x(\x_0, t),\x_0, t)$ and pdf $\wp(\x(\x_0, t) | \x_0, t)$ based on the initial position $\x_0$ \cite{Monin_Yaglom_1971a, Spurk_1997}. These give a Lagrangian form of the Reynolds transport theorem:
\begin{align}
\begin{split}
&d \iiint\limits_{\Omega(\x_0,t)} \wa \, d^3 \x 
\\
&=   \biggl[ \iiint\limits_{\Omega(\x_0,t)} \nabla_{\x_0} \wa \; d^3 \x 
+ \oiint\limits_{\partial \Omega(\x_0,t)} \wa \, \J^\top \cdot  d^2 \x \biggr] \cdot d \x_0
+ \biggl[ \iiint\limits_{\Omega(\x_0,t)} \frac{\partial \wa}{\partial t} \; d^3 \x 
+ \oiint\limits_{\partial \Omega(\x_0,t)} \wa \, \u  \cdot  d^2 \x \biggr] dt
\\
&=   \biggl[ \iiint\limits_{\Omega(\x_0,t)} \bigl( \nabla_{\x_0} \wa + \nabla_{\x} \cdot ( \wa \, \J^\top ) \bigr) d^3 \x \biggr] \cdot d \x_0 
+  \biggl[ \iiint\limits_{\Omega(\x_0,t)} \biggl( \frac{\partial \wa}{\partial t} + \nabla_{\x} \cdot ( \wa \, \u ) \biggr) d^3 \x \biggr] dt,
\end{split}
\label{eq:ReTT_Lagr}
\end{align}
where $\Omega(\x_0,t)$ is the domain, $\nabla_{\x_0}$ is based on $\x_0$ and $\J=\nabla_{\x_0} \x$.  Setting $\wa=\wp$ gives the Liouville equation system:
\begin{align}
\begin{cases}
\nabla_{\x_0} \wp   +  \nabla_{\x} \, \cdot ( \wp \, \J^\top )     = \vect{0} \\
\dfrac{\partial \wp}{\partial t}   +  \nabla_{\x} \cdot (\wp \, \u)=0
\end{cases},
\label{eq:Liouville_Lagr}
\end{align}
which can be summarized as:
\begin{align}
\delfour_{\x_0} \wp   +  \nabla_{\x} \, \cdot ( \wp \, \widetilde{\J}^\top )     = \vect{0},
\label{eq:Liouville_Lagr4D}
\end{align}
where $\widetilde{\J}= \delfour_{\x_0} \x$. This is induced by the system  $\widetilde{\J}^\top=\tens{\mathcal{F}}(\x)$. 
The Liouville equation \eqref{eq:Liouville_Lagr4D} and moment $\lgenn {\obsss}  \rgenn$ then give the Lagrangian maps $\wp(\x | \x_0,t)=\hat{{P}}_{\x_0,t} \, \wp(\x | \0,0)$ and $\obsss(\x,{\x}_0,t)=\hat{{K}}_{\x_0,t} \, \obsss(\x,\0,0)$, invoking the Lagrangian operators 
$\hat{P}_{\x_0,t} =  \exp([\x_0,t] \cdot  \hat{\vec{L}}_{\x_0,t})$, 
$\hat{\vec{L}}_{\x_0,t} \, \wp= - \nabla_{\x} \cdot (  \wp \, \widetilde{\J}^\top ) $ and 
$\hat{{K}}_{\x_0,t}$.

\subsection{Spectral flow systems} These have the generalized density $\hphi(\hu(\vkappa,t), \vkappa,t)$ and pdf $\hf(\hu(\vkappa,t) | \vkappa,t)$, where $\hu$ is the modal amplitude vector and $\vkappa$ the wavenumber vector \cite[e.g.,][]{Pope_2000}. This gives the spectral Reynolds transport theorem: 
\begin{align}
\begin{split}
&d \iiint\limits_{\D(\vkappa,t)} \hphi \, d^3 \hu
\\
&=   \biggl[ \iiint\limits_{\D(\vkappa,t)} \nabla_{\vkappa} \hphi \; d^3 \hu + \oiint\limits_{\partial \D(\vkappa,t)} \hphi \, \tens{\Lambda}^\top \cdot  d^2 \hu \biggr] \cdot d \vkappa
+ \biggl[ \iiint\limits_{\D(\vkappa,t)} \frac{\partial \hphi}{\partial t} \; d^3 \hu + \oiint\limits_{\partial \D(\vkappa,t)} \hphi \, \dhu  \cdot  d^2 \hu \biggr] dt
\\
&=   \biggl[ \iiint\limits_{\D(\vkappa,t)} \bigl( \nabla_{\vkappa} \hphi + \nabla_{\hu} \cdot ( \hphi \, \tens{\Lambda}^\top ) \bigr) d^3 \hu \biggr] \cdot d \vkappa
+  \biggl[ \iiint\limits_{\D(\vkappa,t)} \biggl( \frac{\partial \hphi}{\partial t} + \nabla_{\hu} \cdot ( \hphi \, \dhu ) \biggr) d^3 \hu \biggr] dt,
\end{split}
\label{eq:ReTT_Fourier}
\end{align}
where $\D(\vkappa,t)$ is the domain, $\tens{\Lambda}= \nabla_{\vkappa} \hu$ and $\dhu={\partial \hu}/{\partial t}$.  Setting $\hphi=\hf$ gives the Liouville equation system:
\begin{align}
\begin{cases}
\nabla_{\vkappa} \hf   +  \nabla_{\hu} \, \cdot ( \hf \, {\tens{\Lambda}}^\top )     = \vect{0} \\
\dfrac{\partial \hf}{\partial t}   +  \nabla_{\hu} \cdot ( \hf \, \dhu )=0
\end{cases},
\label{eq:Liouville_Fourier}
\end{align}
Taking $\widetilde{\tens{\Lambda}}=\delfour_{\vkappa} \hu$, this is induced by the system $\widetilde{\tens{\Lambda}}^\top=\tens{\mathcal{F}}(\hu)$. 
Eq.\ \eqref{eq:Liouville_Fourier} and the spectral moment $\lgenn {\obsss}  \rgenn$ then give the maps $\hf(\hu | \vkappa,t)=\hat{{P}}_{\vkappa,t} \, \hf(\hu | \0,0)$ and $\obsss(\hu, \vkappa,t)=\hat{{K}}_{\vkappa,t} \, \obsss(\hu, \0,0)$, invoking spectral operators
$\hat{P}_{\vkappa,t} =  \exp([\vkappa,t]  \hat{\vec{L}}_{\vkappa,t})$, 
$\hat{\vec{L}}_{\vkappa,t}\, \hf = - \nabla_{\hu} \cdot ( \hf \,\tens{\mathcal{F}}(\hu) )$ and 
$\hat{{K}}_{\vkappa,t}$. Many other spectral transformations of dynamical systems, in space and/or time, can be analysed in a similar manner.

\subsection{Coupled chemical reaction and flow systems} These can be described by the generalized density $\chemdens(\m(\x,t),$ $\u(\x,t), \x,t)$ and pdf $\chempdf(\m(\x,t), \u(\x,t) | \x,t)$, where $\m$ is a $k$-dimensional vector of mass (or molar) concentrations of different chemical species. This gives the Reynolds transport theorem:
\begin{align}
\begin{split}
d \int\limits_{\Omega(\x,t)} \chemdens \, d^k \m \, d^3 \u  
&=   \biggl[ \int\limits_{\Omega(\x,t)} \nabla_{\x} \chemdens \; d^k \m \, d^3 \u 
+ \oint\limits_{\partial \Omega(\x,t)} \chemdens \, \M^\top \cdot  d^{k-1} \m 
+ \oint\limits_{\partial \Omega(\x,t)} \chemdens \, \vgrad^\top \cdot   d^2 \u \biggr] \cdot d \x
\\
&\hspace{20pt}+   \biggl[ \int\limits_{\Omega(\x,t)} \frac{\partial \chemdens}{\partial t} \; d^k \m \, d^3 \u 
+ \oint\limits_{\partial \Omega(\x,t)} \chemdens \, \dot{\m} \cdot  d^{k-1} \m 
+ \oint\limits_{\partial \Omega(\x,t)} \chemdens \, \dot{\u} \cdot   d^2 \u \biggr]  d t
\\
&=   \biggl[ \int\limits_{\Omega(\x,t)} \Bigl( \nabla_{\x} \chemdens \;  
+  \nabla_{\m}  \cdot (\chemdens \, \M^\top)  
+ \nabla_{\u}  \cdot (\chemdens \, \vgrad^\top)  \Bigr) d^k \m \, d^3 \u \biggr] \cdot d \x
\\
&\hspace{20pt}+   \biggl[ \int\limits_{\Omega(\x,t)} \Bigr( \frac{\partial \chemdens}{\partial t} 
+ \nabla_{\m}  \cdot (\chemdens \, \dot{\m})
+   \nabla_{\u}  \cdot ( \chemdens \,\dot{\u}) \Bigr) d^k \m \, d^3 \u  \biggr]  d t
\end{split}
\label{eq:ReTT_chemflow}
\raisetag{50pt}
\end{align}
where $\Omega(\x,t)$ is the domain, $\nabla_{\m}=\partial/\partial \m$, $\M = \nabla_{\x}\m$ and $\dot{\m} = \partial \m/\partial t$. Setting $\chemdens=\chempdf$ yields the joint Liouville equations:
\begin{align}
\begin{cases}
\nabla_{\x} \chempdf   +  \nabla_{\m}  \cdot ( \chempdf \, \M^\top )     +  \nabla_{\u}  \cdot ( \chempdf \, \vgrad^\top )   = \vect{0} \\ 
\dfrac{\partial \chempdf}{\partial t}   +  \nabla_{\m} \cdot (\chempdf \, \dot{\m} ) + \nabla_{\u} \cdot (\chempdf \, \dot{\u})=0
\end{cases},
\label{eq:Liouville_chemflow}
\end{align}
which are induced by $\delfour_{\x} [\m,\u]=\tens{\mathcal{F}}(\vect{m}, \u)$. 
From the moment $\lgenn {\obsss}  \rgenn$ based on the observable $\obsss(\m, \u, \x,t)$, 
these in turn give the maps $\chempdf(\m, \u | \x,t)=\hat{{P}}_{\x,t} \, \chempdf(\m, \u | \0,0)$ and $\obsss(\m, \u,{\x},t)=\hat{{K}}_{\x,t} \, $ $\obsss (\m, \u,\0,0)$, invoking spatiotemporal Perron-Frobenius and Koopman operators similar to those in part (\ref{sect:spatiotempvel}), now with Liouville operator 
$\hat{\vec{L}}_{\x,t}\, \chempdf= - \nabla_{\m, \u} \cdot (\chempdf \, \tens{\mathcal{F}}(\vect{m}, \u)) $. 
These relations give a very different approach for the probabilistic analysis of chemical reaction dynamical systems (c.f. \cite{Gorban_etal_2004}).

\subsection{Chemical reaction-dependent flow systems} In the probabilistic description these are described by the probability $r$-form $\rho^r = \rho^r_{\x,\m,t}$, a function of the local velocity $\u(\x,\m,t)$ and conditioned on $\x$, $\m$ and $t$. Eq.\ \eqref{eq:Liouville_diffform_aug} gives the Liouville equation system:
\begin{align}
 \mathcal{L}_{\u \comp (\x,\m,t)}^{(\x, \m,t )} \rho^r 
=  \nabla_{\x,\m,t} \; \rho^r + i_{\u \comp (\x,\m,t)}^{(\x, \m,t)}  \, d \rho^r + d ( i_{\u \comp (\x,\m,t)}^{(\x,\m,t)} \, \rho^r)
= \vec{0}.
\label{eq:Liouville_chemflow2}
\end{align}
For global velocity coordinates and pdf $\chempdf(\u | \x, \m,t)$, these reduce to analogs of \eqref{eq:Liouville_uxt} and also $\nabla_{\m} \chempdf   +  \nabla_{\u} \, \cdot ( \chempdf \, \K^\top ) = \vect{0} $ with $\K =  \nabla_{\m} \u$ for variations in chemical species, induced by the system $[\nabla_{\x}, \nabla_{\m},\partial/\partial t]^\top \u= \delfour_{\x,\m} \u = \tens{\mathcal{F}}(\u)$. From the moment $\lgenn {\obsss}  \rgenn$ based on the observable $\obsss(\u, \x, \m,t)$, these give the spatiochemicotemporal maps $\chempdf(\u | \x,\m, t)= $ $\hat{{P}}_{\x,\m, t} \, $ $\chempdf(\u | \0,\0,0)$ and $\obsss(\u,{\x},\m, t)=$ $\hat{{K}}_{\x,\m, t}$ $\obsss(\u,\0,\0,0)$, 
invoking the spatiochemicotemporal operators 
$\hat{P}_{\x,\m,t} =  \exp([\x,\m,t] \,  \hat{\vec{L}}_{\x,\m,t})$, 
$\hat{\vec{L}}_{\x,\m,t} \chempdf$ $= - \nabla_{\u} \cdot ( \chempdf \, \tens{\mathcal{F}}(\u)) $ and 
$\hat{{K}}_{\x,\m,t}$.



\section{Conclusions} \label{sect:concl} 

In this review, we present a unified framework for the derivation of Rey\-nolds transport theorems, Liouville equations and Perron-Frobenius and Koopman operators in different spaces, each of which provides a continuous map between different points or domains in these spaces, described by the integral curves of a vector or tensor field. These extend the well-known volumetric-temporal formulations of these theorems and operators, firstly to a velocimetric-spatial formulation, and then to more general parametric formulations. 
The velocimetric-spatial formulation provides spatial maps between positions in a time-independent flow field, connected by a velocity gradient tensor field, while the parametric formulations provide parametric maps between positions in a manifold, connected by a vector or tensor field. 
The most general parametric formulation invokes multivariate extensions of several exterior calculus operators including the flow, pullback, pushforward, Lie derivative and interior product, and the concept of a probability differential $r$-form. 
The analyses reveal the existence of multivariate continuous (Lie) symmetries  -- in time, space and/or parametric coordinates -- induced by a vector or tensor field associated with a conserved quantity, which will be manifested in the form of subsidiary integral conservation laws. These findings significantly expand the scope of available methods for the reduction of fluid flow and dynamical systems. 

To demonstrate their insights and breadth, the findings are used to present new formulations of these theorems and operators for several prominent case study systems in fluid mechanics and dynamical systems. These include spatial (time-independent) and spatiotemporal fluid flows, flow systems with pairwise or $n$-wise correlations, phase space systems, Lagrangian flows, spectral flows, and systems with coupled chemical reaction and flow processes. 

This study opens a number of important avenues for further research. The new formulations of the Reynolds transport theorem \eqref{eq:Re_tr_u_3D}, \eqref{eq:Re_tr_diffform}, \eqref{eq:Re_tr_diffform_aug} and \eqref{eq:Re_tr_vect_calc_gen} reveal the existence of alternative formulations of the known integral conservation laws based on different fluid densities, which require more detailed examination.
The connections between the multivariate continuous (Lie) symmetries revealed in \S\ref{sect:general}, and recent analyses of one-parameter Lie symmetries of the Reynolds transport theorem and related conservation laws \cite[e.g.,][]{Haltas_Ulusoy_2015, Ercan_Kavvas_2015, Ercan_Kavvas_2017}, as well as to other diffeomorphisms \cite{Bowen_Ruelle_1975, Ruelle_1976}, warrant further study. Furthermore, the connections between spatiotemporal Lie symmetries and coherent structures, for example as identified in the Navier-Stokes equations \cite[e.g.,][]{Sharma_etal_2016}, require further exploration. 
Finally, the new spatial and parametric Liouville equations and Perron-Frobenius and Koopman operators presented in \S\ref{sect:spat_anal}-\ref{sect:general} offer an assortment of new tools for the analysis of a wide variety of fluid flow and dynamical systems. 

\section*{Funding}
This work was funded by the Australian Research Council Discovery Projects grant DP140104402, and by Institute Pprime, CNRS, (former) R\'egion Poitou-Charentes and l'Agence Nationale de la Recherche Chair of Excellence (TUCOROM), all in Poitiers, France, and CentraleSup\'elec, Gif-sur-Yvette, France.

\section*{Acknowledgments}
This work was largely completed during two periods of sabbatical leave by RN in France in 2014 and 2018, supported by UNSW and French funding sources. We thank Daniel Bennequin, Juan Pablo Vigneaux and the group at IMJ-PRG, Universit\'e Paris Diderot, as well as Andreas Spohn, Markus Abel, Nicolas H\'erouard, Tony Kennedy and Ali Ghaderi for comments on early versions of the manuscript, Stephen Whitaker for a difficult-to-find reference, and Hayward Maberley of the UNSW Canberra library for literature support.

\newpage

\appendix
\renewcommand{\thesection}{Appendix \Alph{section}}
\renewcommand{\thesubsection}{\Alph{section}.\arabic{subsection}}
\renewcommand{\thefigure}{\Alph{section}.\arabic{figure}}
\renewcommand{\theequation}{\Alph{section}.\arabic{equation}}

\section{\label{sect:Apx_Re_temp} Proofs of the Volumetric-Temporal Reynolds Transport Theorem} 

We first revisit two proofs of the extended form of the volumetric-temporal Reynolds transport theorem \eqref{eq:Re_tr}, based respectively on continuum mechanics \cite{Reynolds_1903, Prager_1961, Sokolnikoff_Redheffer_1966, White_1986, Bear_Bachmat_1991, Tai_1992, Streeter_etal_1998, Fox_etal_2004, Leal_2007, Munson_etal_2010} or Lagrangian coordinate transformation \cite{Aris_1962, Slattery_1972, Flanders_1973, Spurk_1997, Slattery_1999, Dvorkin_Goldschmidt_2005}.

\subsection{Proof 1: Continuum Mechanics}

For the first proof, we adopt the continuum assumption and Eulerian description of fluid flow, in which the generalized density $\alpha(\x,t)$ of a conserved property in volumetric space can be represented as a function of Cartesian position $\x = [x,y,z]^\top$ and time $t$ within a prescribed geometric region (control volume) as the fluid moves past.  We also consider the Lagrangian or material description of fluid flow, in which each fluid element is assigned a characteristic label, for example its position vector $\x_0$ at time $t_0$. The position of each fluid element at time $t$ is then $\x(\x_0,t)$, giving the fluid element velocity ${\u}^{L}(\x_0, t) = \partial \x(\x_0,t)/\partial t$. The two descriptions can then be united by the equivalence of the material and Eulerian velocities \cite{Spurk_1997, Durst_2008}:
\begin{align}
 \frac{\partial \x(\x_0, t)}{\partial t} = {\u}^L (\x_0, t) = \u(\x,t).
 \label{eq:Eul_Lagr}
\end{align}
The two descriptions in \eqref{eq:Eul_Lagr} also establish the equivalence of the substantial and total derivatives, here based on the intrinsic fluid velocity:
\begin{align}
\frac{D \alpha(\x,t)}{Dt} 
=\frac{\partial \alpha}{\partial t} + \nabla_{\x} \alpha \cdot \u (\x,t)
=\frac{\partial \alpha}{\partial t} + \nabla_{\x} \alpha \cdot \frac{\partial \x(\x_0,t)}{\partial t}
= \frac{d \alpha(\x,t)}{d t} .
\label{eq:total_deriv1}
\end{align}

We here adopt the standard treatment used in fluid mechanics, to consider the motion of a contiguous body of fluid $\Omega(t)$ -- commonly referred to as the {\it fluid volume} -- within its surrounding volumetric space. In the Eulerian description, the fluid volume is considered to be in motion with local velocity field $\u(\x, t)$ through a prescribed region of interest, known as the {\it control volume}. For the present analysis, we also consider a moving and smoothly deforming control volume with local velocity field $\u_{CV}$. A schematic diagram of this situation is shown in Figure \ref{fig:moving_CV}. 
\begin{figure}[h]
 \begin{picture}(400,110)
 \put(80,0){ \includegraphics[height=100pt]{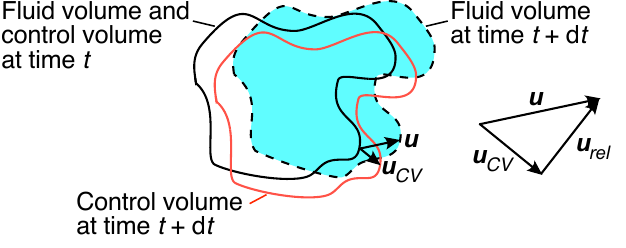}}
 \end{picture}
\caption{Schematic diagram of the motions of a fluid volume and a moving control volume (redrawn after 
\citep[][Fig. 3.5]{White_1986} or \citep[][Fig. 4.21]{Munson_etal_2010}).}
\label{fig:moving_CV}
\end{figure}
As shown, at time $t$, both the fluid volume and control volume occupy the position shown with a black border. By time $t + dt$, the fluid volume has moved to the position shown in blue with a dashed border (here drawn without any change in shape), while the control volume has moved to the position shown with a red border (also drawn without change in shape). We are interested in the velocity of the fluid relative to the moving control volume, here denoted $\u_{rel}$. From the vector diagram shown -- for the Cartesian coordinate system used here -- it is evident that $\u_{CV}+\u_{rel}=\u$, hence:
\begin{equation}
\u_{rel}=\u-\u_{CV}.
\label{eq:rel_vel}
\end{equation}
Eq.\ \eqref{eq:rel_vel} is implied or derived in several fluid mechanics references that consider moving control volumes \citep[e.g.][]{White_1986, Streeter_etal_1998, Fox_etal_2004, Munson_etal_2010}, including the original analysis by Reynolds \cite{Reynolds_1903}. For curvilinear coordinate systems, an extended relation has also been presented \cite{Dvorkin_Goldschmidt_2005}. For a stationary control volume $\u_{CV}=\0$, \eqref{eq:rel_vel} simplifies to $\u_{rel} = \u$.  We further see that even if both the fluid volume and control volume are smoothly deforming, \eqref{eq:rel_vel} remains unchanged, except that the control volume velocity $\u_{CV}$ is no longer constant but becomes a velocity field $\u_{CV}(\x,t)$ \cite{White_1986, Munson_etal_2010}.

Consider the motion of the fluid relative to the moving and smoothly deforming control volume, for which a schematic diagram of several streamlines is given in Figure \ref{fig:boundary_els_vol}(a).  
\begin{figure}[h]
\begin{center}
 \begin{picture}(400,140)
 \put(50,0){ \includegraphics[height=130pt]{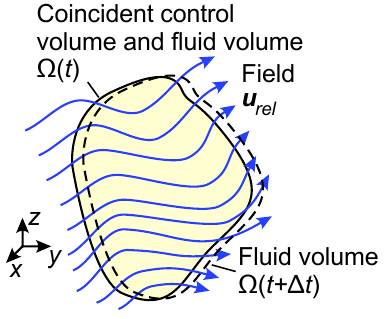}}
 \put(50,0){\small (a)}
 \put(230,0){ \includegraphics[height=100pt]{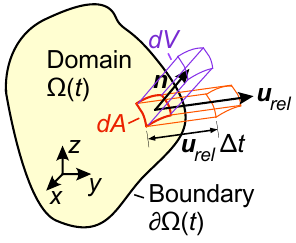} }
 \put(230,0){\small (b)} 
 \end{picture}
\end{center}
\caption{Schematic diagrams of (a) a velocity field for fluid flow relative to a geometric control volume, and (b) a volume element on the boundary induced by the flow.}
\label{fig:boundary_els_vol}
\end{figure}
For the moving frame of reference, it is necessary to redefine the substantial and total derivatives in terms of the relative velocity field, as given in \eqref{eq:total_deriv} \cite{Truesdell_Toupin_1960, Slattery_1972, Dvorkin_Goldschmidt_2005}:
\begin{align}
\frac{D \alpha(\x,t)}{Dt} 
=\frac{\partial \alpha}{\partial t} + \nabla_{\x} \alpha \cdot \u_{rel} (\x,t)
=\frac{\partial \alpha}{\partial t} + \nabla_{\x} \alpha \cdot \biggl(\frac{\partial \x(\x_0,t)}{\partial t}\biggr)_{rel}
= \frac{d \alpha(\x,t)}{d t} .
\label{eq:total_deriv2}
\end{align}
We examine the total conserved quantity $Q(t)$ in the fluid volume, given by the volumetric integral of the generalized density $\alpha(\x,t)$:
\begin{align}
Q(t) = \iiint\limits_{\Omega(t)} \alpha(\x,t) \, dV,
\end{align}
where 
$dV=d^3\x=dxdydz$ is an infinitesimal volume element. Since $Q(t)$ is a function only of time, it is possible to write its total derivative as:
\begin{align}
\frac {d Q(t) }{dt}
= \frac {d }{dt} \iiint\limits_{\Omega(t)} \alpha(\x,t) \, dV
= \lim\limits_{\Delta t \to 0} \frac{1}{\Delta t} 
\biggl[ \iiint\limits_{\Omega(t + \Delta t)} \alpha(\x,t +\Delta t) \, dV - \iiint\limits_{\Omega(t)} \alpha(\x,t) \, dV \biggr].
\label{eq:Re_proof1}
\end{align}
The second form follows from the definition of the derivative, with $\Omega(t + \Delta t)$ here interpreted as the fluid volume (relative to the moving control volume) at time $t + \Delta t$. By a Taylor expansion \cite{Tai_1992}:
\begin{align}
\alpha(\x,t +\Delta t) = \alpha(\x,t) + \frac{\partial \alpha(\x,t) }{\partial t} \Delta t + \frac{1}{2} \frac{\partial^2 \alpha(\x,t) }{\partial t^2} (\Delta t)^2 + ...
\end{align}
Substitution into \eqref{eq:Re_proof1} gives
\begin{align}
\begin{split}
&\frac {d Q(t) }{dt}
\\&= \lim\limits_{\Delta t \to 0} \frac{1}{\Delta t} 
\biggl[ \iiint\limits_{\Omega(t + \Delta t)} 
\biggl( 
\alpha(\x,t) + \frac{\partial \alpha(\x,t) }{\partial t} \Delta t + \frac{1}{2} \frac{\partial^2 \alpha(\x,t) }{\partial t^2} (\Delta t)^2 + ...
\biggr) \, dV - \iiint\limits_{\Omega(t)} \alpha(\x,t) \, dV \biggr]
\\
&= \lim\limits_{\Delta t \to 0} \frac{1}{\Delta t} 
\iiint\limits_{\Omega(t +\Delta t)}
\frac{\partial \alpha(\x,t) }{\partial t} \Delta t \, dV
+
\lim\limits_{\Delta t \to 0} \frac{1}{\Delta t} 
\biggl[ \iiint\limits_{\Omega(t + \Delta t)} \alpha(\x,t)  dV - \iiint\limits_{\Omega(t)} \alpha(\x,t) \, dV \biggr]
\\
&= 
\iiint\limits_{\Omega(t )}
\frac{\partial \alpha(\x,t) }{\partial t}  \, dV
+
\lim\limits_{\Delta t \to 0} \frac{1}{\Delta t} \;
\iiint\limits_{\Omega(t + \Delta t) - \Omega(t)} \alpha(\x,t) \, dV .
\end{split}
\label{eq:Re_proof2}
\end{align}
Note the second-order and higher derivatives vanish in the limit.

We see that the second integral reduces to that of a thin domain (of variable sign) adjacent to the boundary, created by the relative motion of the fluid between $t$ and $t+\Delta t$. To examine this, consider a volume element $dV$ in this boundary region, illustrated schematically in Figure \ref{fig:boundary_els_vol}(b) \citep[e.g.][]{Prager_1961, Sokolnikoff_Redheffer_1966, White_1986, Bear_Bachmat_1991, Tai_1992, Streeter_etal_1998, Anderson_2001, Fox_etal_2004, Leal_2007, Munson_etal_2010}.  At time $t$, the rate of change of the fluid position relative to the control volume boundary is $(\partial \x(\x_0, t)/\partial t)_{rel} = \u_{rel} $. In time $\Delta t$, this will induce the displacement $\u_{rel} \Delta t$ in the direction of $\u_{rel}$. The volumetric element $dV$ is therefore the inclined cylinder formed by projection of the boundary element $dA$ over the inclined distance $\u_{rel} \Delta t$, accounting for its height in the direction of the outward unit normal $\n$. This gives the intrinsic length $d \ell = \u_{rel} \Delta t \cdot \n$, hence $dV = d \ell dA = \u_{rel} \Delta t \cdot \n \, dA$. For flow out of the control volume, this will be positive, while for inwards flow, this will be negative. Thus \eqref{eq:Re_proof2} reduces to
\begin{align}
\begin{split}
\frac {d Q(t) }{dt}
&= 
\iiint\limits_{\Omega(t )}
\frac{\partial \alpha(\x,t) }{\partial t}  \, dV
+
\lim\limits_{\Delta t \to 0} \frac{1}{\Delta t} \;
\oiint\limits_{\partial \Omega(t)} \alpha(\x,t) \, \u_{rel} \Delta t \cdot \n \, dA
\\
&= 
\iiint\limits_{\Omega(t )}
\frac{\partial \alpha(\x,t) }{\partial t}  \, dV
+
\oiint\limits_{\partial \Omega(t)} \alpha(\x,t) \, \u_{rel} \cdot \n \, dA,
\end{split}
\label{eq:Re_proof3}
\end{align} 
where $\partial \Omega(t)$ is the domain boundary. 
This directly gives the first form of the extended Reynolds transport theorem in \eqref{eq:Re_tr}. The second form in \eqref{eq:Re_tr} is obtained by the divergence theorem. $\blacksquare$

\subsection{Proof 2: Lagrangian Coordinate Transformation}

For the second proof, we follow a Lagrangian description \cite{Flanders_1973,Spurk_1997} and consider Lagrangian coordinates $\x_0=[x_0,y_0,z_0]^\top$ with the fixed original domain $\Omega(t_0)$. Rewriting the first two parts of \eqref{eq:Re_proof1} in Lagrangian coordinates gives
\begin{align}
\frac {d Q(t) }{dt}
= \frac {d }{dt} \iiint\limits_{\Omega(t)} \alpha(\x,t) \, dV
= \frac {d }{dt} \iiint\limits_{\Omega(t_0)} \alpha(\x(\x_0,t),t) \biggl | \frac{\partial \x}{\partial \x_0} \biggr | dV_0,
\label{eq:Re_proof_B1}
\end{align}
where $V_0=dx_0 dy_0 dz_0$ and $ | {\partial \x}/{\partial \x_0}  |$ is the determinant of the Jacobian matrix of Eulerian with respect to Lagrangian coordinates. The domain in the last part of \eqref{eq:Re_proof_B1} is now independent of time, so the derivative can be brought inside the integral. Furthermore, since fluid elements are unique and indivisible, the Jacobian ${\partial \x}/{\partial \x_0}$ will be everywhere non-singular. The derivative of the determinant, in the moving frame of reference, is \cite{Flanders_1973, Zwillinger_2003}:
\begin{align}
\frac{d}{dt} \biggl | \frac{\partial \x}{\partial \x_0} \biggr | 
= \biggl | \frac{\partial \x}{\partial \x_0} \biggr | \, \nabla_{\x} \cdot \biggl( \frac{\partial \x}{\partial t} \biggr)_{rel}
= \biggl | \frac{\partial \x}{\partial \x_0} \biggr | \, \nabla_{\x} \cdot \u_{rel}.
\label{eq:det_rel}
\end{align}
Expanding \eqref{eq:Re_proof_B1} using \eqref{eq:total_deriv2} and \eqref{eq:det_rel}, and then reverting back to the variable domain, gives:
\begin{align}
\begin{split}
\frac {d Q(t) }{dt}
&=  \iiint\limits_{\Omega(t_0)} 
\biggl [ \biggl( \frac {\partial \alpha }{\partial t} + (\nabla_{\x} \alpha) \cdot \u_{rel}
\biggr) \biggl | \frac{\partial \x}{\partial \x_0} \biggr | 
+ \alpha \biggl | \frac{\partial \x}{\partial \x_0} \biggr | \, \nabla_{\x} \cdot \u_{rel}
\biggr ] dV_0
\\
&=  \iiint\limits_{\Omega(t)} 
\biggl [ \frac {\partial \alpha }{\partial t} +(\nabla_{\x} \alpha) \cdot \u_{rel}
+ \alpha  \nabla_{\x} \cdot \u_{rel}
\biggr ] dV.
\end{split}
\label{eq:Re_proof_B2}
\end{align}
This is identical to the second form of the extended Reynolds transport theorem in \eqref{eq:Re_tr}. The first form 
is then obtained by the divergence theorem. $\blacksquare$




\section{\label{sect:Apx_Re_spatial} Proofs of the Velocimetric-Spatial Reynolds Transport Theorem} 
\setcounter{equation}{0}

We now provide two proofs of the velocimetric-spatial Reynolds transport theorem  \eqref{eq:Re_tr_u_3D} for time-independent flows, based respectively on arguments from continuum mechanics and a coordinate transformation method. These follow the essential details of the volumetric-temporal proofs in   \ref{sect:Apx_Re_temp}. 

\subsection{Proof 1: Continuum Mechanics (Steady Flow)}

\begin{figure}[h]
\begin{center}
 \begin{picture}(400,130)
 \put(50,0){ \includegraphics[height=130pt]{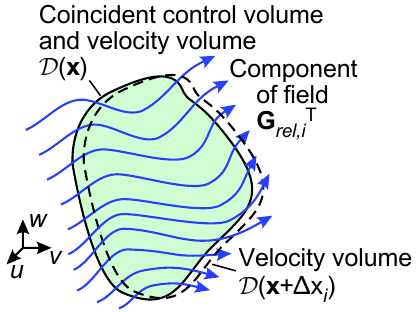}}
 \put(50,0){\small (a)}
 \put(230,0){ \includegraphics[height=100pt]{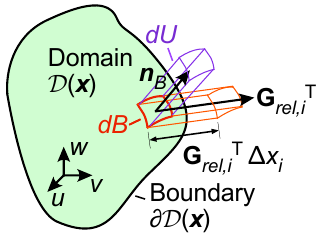} }
 \put(230,0){\small (b)} 
 \end{picture}
\end{center}
\caption{Schematic diagrams showing the $i$th component of (a) a velocity gradient field for steady flow relative to a velocimetric control volume, and (b) a velocity volume element on the domain boundary induced by the tensor field.}
\label{fig:boundary_els_vel}
\end{figure}

For the first proof we again make the continuum assumption, and consider an Eulerian phase space (volumetric and velocimetric) description of fluid flow, in which the density $\varphi(\u,\x)$ of a conserved quantity in velocimetric space can be represented as a  function of velocity $\u = [u,v,w]^\top$ and position $\x = [x,y,z]^\top$, for which $[\u,\x]^\top$ gives a Cartesian coordinate system. In this description, $\varphi(\u,\x)$ can be interpreted physically as the conserved quantity carried per unit of velocity and volume space by a fluid element with a velocity between $\u$ and $\u + d\u$ at the position between $\x$ and $\x + d\x$. This representation assumes time-independent flow, for example a steady velocity field. 
We also consider an alternative description in which the velocity at each point $\u^R(\u_0, \x)$ is a function of the velocity $\u_0$ at some reference location $\x_0$, for which the spatial coordinates $\x$ are independent variables. The two descriptions are united by the equivalence of the velocity gradient tensor:
\begin{align}
\frac{\partial \u^R}{\partial \x} (\u_0, \x) = \vgrad (\u,\x) .
\end{align}
As with the temporal analysis, we also incorporate a spatially varying velocimetric control volume which undergoes a changing reference velocity gradient $\vgrad_{CV}$, giving the relative gradient  $\vgrad_{rel} = \vgrad - \vgrad_{CV}$.  
A set of field lines for such a system -- for example for one spatial coordinate of the tensor -- is illustrated schematically in Figure \ref{fig:boundary_els_vel}(a). 
 
Now consider the integral of the generalized phase space density $\varphi(\u,\x)$ over the velocity domain:
\begin{align}
\F(\x)= \iiint\limits_{\D(\x)} \varphi(\u,\x) \, dU,
\label{eq:bint_vel}
\end{align}
where $dU=d^3\u =dudvdw$ is the velocity volume element. $\F(\x)$ corresponds to the total conserved quantity (integrated over the velocity space) per unit volume at the position between $\x$ and $\x + d\x$, i.e., it is equivalent to the generalized volumetric density $\alpha(\x)$, in this case for a time-independent system. 
Since $\F(\x)$ is multivariate, it is not possible to define the total derivative, but we can directly consider its differential:
\begin{align}
 d \F(\x) 
= \nabla_{\x} \F(\x) \cdot d\x
= \sum\limits_{i=1}^3  \frac{\partial \F(\x)}{\partial x_i}  d x_i
= \sum\limits_{i=1}^3  \frac{\partial}{\partial x_i} \biggl[ \iiint\limits_{\D(\x)} \varphi(\u,\x) \, dU \biggr] d x_i,
\label{eq:Re_proof0_vel}
\end{align}
using $\x = [x_1,x_2,x_3]^\top$. 
Each partial derivative is, by definition:
\begin{align}
\frac{\partial \F(\x)}{\partial x_i} 
= \lim\limits_{\Delta x_i \to 0} \frac{1}{\Delta x_i} 
\biggl[ \iiint\limits_{\D(\x + \Delta x_i)} \varphi(\u,\x +\Delta x_i) \, dU - \iiint\limits_{\D(\x)} \varphi(\u,\x) \, dU \biggr],
\label{eq:Re_proof1_vel}
\end{align}
where we use the notation $(\x + \Delta x_i)$ to indicate $(x+\Delta x, y,z)$, $(x, y + \Delta y ,z)$ or $(x, y,z+\Delta z)$ respectively for $x_i \in [x_1,x_2,x_3]$. $\D(\x + \Delta x_i)$ is then interpreted as the velocity domain shifted to position $\x + \Delta x_i$. By a one-dimensional Taylor expansion -- or a multi-dimensional expansion with non-zero translation in only one coordinate -- we obtain \cite{Tai_1992}:
\begin{align}
\varphi(\u,\x +\Delta x_i) = \varphi(\u,\x) + \frac{\partial \varphi(\u,\x) }{\partial x_i} \Delta x_i + \frac{1}{2} \frac{\partial^2 \varphi(\u,\x) }{\partial x_i^2} (\Delta x_i)^2 + ...
\label{eq:Taylor_x}
\end{align}
Substitution into \eqref{eq:Re_proof1_vel} gives
\begin{align}
\begin{split}
&\frac{\partial \F(\x)}{\partial x_i} 
\\&=
\lim\limits_{\Delta x_i \to 0} \frac{1}{\Delta x_i} 
\biggl[ \iiint\limits_{\D(\x + \Delta x_i)}
\biggl( 
\varphi(\u,\x) + \frac{\partial \varphi(\u,\x) }{\partial x_i} \Delta x_i + \frac{1}{2} \frac{\partial^2 \varphi(\u,\x) }{\partial x_i^2} (\Delta x_i)^2 + ...
\biggr) \, dU 
\\& \hspace{20pt} - \iiint\limits_{\D(\x)} \varphi(\u,\x) \, dU  \biggr]
\\
&= \lim\limits_{\Delta x_i \to 0} \frac{1}{\Delta x_i} 
 \iiint\limits_{\D(\x + \Delta x_i)}
\frac{\partial \varphi(\u,\x) }{\partial x_i} \Delta x_i  \, dU
+
\lim\limits_{\Delta x_i \to 0} \frac{1}{\Delta x_i} 
\biggl[ \iiint\limits_{\D(\x + \Delta x_i)} \varphi(\u,\x)  dU - \iiint\limits_{\D(\x)} \varphi(\u,\x) \, dU \biggr]
\\
&= 
\iiint\limits_{\D(\x )}
\frac{\partial \varphi(\u,\x) }{\partial x_i}  \, dU
+
\lim\limits_{\Delta x_i \to 0} \frac{1}{\Delta x_i} 
\iiint\limits_{\D(\x + \Delta x_i) - \D(\x)} \varphi(\u,\x) \, dU,
\end{split}
\label{eq:Re_proof2_vel}
\end{align}
where again the second-order and higher derivatives vanish in the limit. 

We again see that the second integral reduces to that of a thin domain (of variable sign) in velocimetric space adjacent to the boundary, created by translation of the field (relative to the domain) between $\x$ and $\x+\Delta x_i$. Consider a velocimetric element $dU$ in this boundary region, illustrated schematically in Figure \ref{fig:boundary_els_vel}(b), with spatial displacement in only one component $x_i \in \{x_1,x_2,x_3\}$.  At position $\x$, the velocity gradient relative to the boundary $\partial \D(\x)$ is $\vgrad_{rel,i} = \partial \u_{rel}/\partial x_i$, for consistency here taken as a row vector. Over the distance $\Delta x_i$, this will induce the change in velocity $\vgrad_{rel,i}^\top \Delta x_i$ in the direction described by $\vgrad_{rel,i}^\top$. The velocimetric element $dU$ is therefore the inclined cylinder formed by projection of the boundary element $dB$ over the inclined distance $\vgrad_{rel,i}^\top \Delta x_i$, accounting for its height in the direction of the outward unit normal $\n_B$. This gives the intrinsic length $d \ell_B = \vgrad_{rel,i}^\top \Delta x_i \cdot \n_B$, hence $dU = d \ell_B dB = \vgrad_{rel,i}^\top \Delta x_i \cdot \n_B dB$.  Thus \eqref{eq:Re_proof2_vel} reduces to
\begin{align}
\begin{split}
\frac{\partial \F(\x)}{\partial x_i} 
&= 
\iiint\limits_{\D(\x )}
\frac{\partial \varphi(\u,\x) }{\partial x_i}  \, dU
+
\lim\limits_{\Delta x_i \to 0} \frac{1}{\Delta x_i} 
\oiint\limits_{\partial \D(\x)} \varphi(\u,\x) \, \vgrad_{rel,i}^\top \Delta x_i \cdot \n_B \, dB
\\
&= 
\iiint\limits_{\D(\x )}
\frac{\partial \varphi(\u,\x) }{\partial x_i}  \, dU
+
\oiint\limits_{\partial \D(\x)} \varphi(\u,\x) \, \vgrad_{rel,i}^\top  \cdot \n_B \, dB.
\end{split}
\label{eq:Re_proof3_vel}
\end{align} 
Assembling these into \eqref{eq:Re_proof0_vel}, we obtain the differential
\begin{align}
 d \F(\x) 
= \sum\limits_{i=1}^3  \biggl[
\iiint\limits_{\D(\x )}
\frac{\partial \varphi(\u,\x) }{\partial x_i}  \, dU
+
\oiint\limits_{\partial \D(\x)} \varphi(\u,\x) \, \vgrad_{rel,i}^\top \cdot \n_B \, dB
\biggr] d x_i.
\label{eq:Re_proof4_vel}
\end{align}
The divergence theorem can be extended (strictly, in the form of Stokes' theorem) to any metric space \cite[e.g.,][]{Cartan_1945}. Applying its three-dimensional velocimetric formulation 
then gives:
\begin{align}
 d \F(\x) 
= \sum\limits_{i=1}^3  \biggl[
\iiint\limits_{\D(\x )}
\biggl[
\frac{\partial \varphi(\u,\x) }{\partial x_i} 
+
\nabla_{\u} \cdot  \bigl( \varphi(\u,\x) \, \vgrad_{rel,i}^\top  \bigr) 
\biggr] dU
\biggr] d x_i.
\label{eq:Re_proof5_vel}
\end{align} 
Reverting to $d^3 \u = dU$, $d^2 \u=\n_B dB$, and using vector and tensor notation based on the selected gradient convention, \eqref{eq:Re_proof4_vel}-\eqref{eq:Re_proof5_vel} give the velocimetric-spatial Reynolds transport theorem in \eqref{eq:Re_tr_u_3D}. $\blacksquare$

\vspace{10pt}

{\bf Alternative proof:} A more direct proof is to recognize $d \F(\x)$ in \eqref{eq:Re_proof0_vel} as the directional derivative $D_{\r} \, \F(\x) = \r \cdot \nabla_{\x} \, \F(\x)$, in the direction of the differential vector $\r=d\x$.   By definition:
\begin{align}
\begin{split}
D_{d\x} \, \F(\x) 
&= \lim\limits_{h \to 0} \frac{\F(\x + h d\x) - \F(\x)}{h} 
\\
&= \lim\limits_{h \to 0} \frac{1}{h} \biggl [ \iiint\limits_{\D(\x + h d\x)} \varphi(\u,\x+ h d\x) \, dU -  \iiint\limits_{\D(\x)} \varphi(\u,\x) \, dU \biggr].
\end{split}
\label{eq:Re_proof1_dir_vel}
\end{align}
Using a multidimensional Taylor expansion:
\begin{align}
\varphi(\u,\x + h d\x) = \varphi(\u,\x) + h d\x^\top \nabla_{\x} \varphi(\u,\x)     + \frac{h^2}{2} d\x^\top \nabla^2_{\x} \varphi(\u,\x) d\x + ...,
\label{eq:Taylor_x_mult}
\end{align}
where $\nabla_{\x}^2=\nabla_{\x} (\nabla_{\x})^\top$ is the second derivative or Hessian operator, 
we obtain:
\begin{align}
\begin{split}
&D_{d\x} \, \F(\x) 
\\&= \lim\limits_{h \to 0} \frac{1}{h} \biggl [ \iiint\limits_{\D(\x + h d\x)} 
\biggl( \varphi(\u,\x) + h d\x^\top \nabla_{\x} \varphi(\u,\x)     + \frac{h^2}{2} d\x^\top \nabla^2_{\x} \varphi(\u,\x) d\x + ... \biggr)
 \, dU  
 \\& \hspace{20pt}
 -  \iiint\limits_{\D(\x)} \varphi(\u,\x) \, dU \biggr]
 \\
&=
 \lim\limits_{h \to 0} \frac{1}{h}  \iiint\limits_{\D(\x + h d\x)} 
h d\x^\top \nabla_{\x} \varphi(\u,\x)    
 \, dU  
+
 \lim\limits_{h \to 0} \frac{1}{h}  \biggl [\iiint\limits_{\D(\x + h d\x)} 
 \varphi(\u,\x)     
 \, dU -  \iiint\limits_{\D(\x)} \varphi(\u,\x) \, dU \biggr]
 \\
 &=
 \iiint\limits_{\D(\x)} 
d\x^\top \nabla_{\x} \varphi(\u,\x)    
 \, dU  
+
 \lim\limits_{h \to 0} \frac{1}{h}  \iiint\limits_{\D(\x + h d\x)-\D(\x)} 
 \varphi(\u,\x)     
 \, dU ,
 \end{split}
\end{align}
where again the second and higher derivatives vanish. The analysis uses the same directional argument as before, now in resultant form $dU = h d\x^\top \vgrad_{rel} \n_B dB$, giving the limit
\begin{align}
d \F(\x) =
D_{d\x} \, \F(\x) 
&=
 \iiint\limits_{\D(\x)} 
d\x^\top \nabla_{\x} \varphi(\u,\x)    
 \, dU  
+
  \oiint\limits_{\partial \D(\x)}
 \varphi(\u,\x)     
 \, d\x^\top \vgrad_{rel} \n_B dB.
 \label{eq:Re_proof2_dir_vel}
\end{align}
This is identical to \eqref{eq:Re_proof4_vel} and the first part of \eqref{eq:Re_tr_u_3D}. $\blacksquare$

\subsection{Proof 2: Reference Velocity Coordinate Transformation}

For the second proof, we consider the alternative description based on a reference set of velocity coordinates $\u_0=[u_0,v_0,w_0]^\top$ in a spatially fixed velocity domain $\D(\x_0)$. Rewriting the left hand side of \eqref{eq:Re_proof0_vel} gives
\begin{align}
 d \F(\x) 
= \sum\limits_{i=1}^3  \frac{\partial}{\partial x_i} \biggl[ \iiint\limits_{\D(\x)} \varphi(\u,\x) \, dU \biggr] d x_i
= \sum\limits_{i=1}^3  \frac{\partial}{\partial x_i} \biggl[ \iiint\limits_{\D(\x_0)} \varphi(\u(\u_0,\x),\x) \, \biggl | \frac{\partial \u}{\partial \u_0} \biggr |  \, dU_0 \biggr] d x_i,
\label{eq:Re_proof0_vel_B1}
\end{align}
where $dU_0=du_0 dv_0 dw_0$ and $ | {\partial \u}/{\partial \u_0}  |$ is the Jacobian determinant for this coordinate transformation. For the class of time-independent flow systems examined here, we consider the Jacobian ${\partial \u}/{\partial \u_0}$ to be everywhere non-singular. Using the velocity analog of the relation \eqref{eq:det_rel} for independent spatial coordinates $\x$:
\begin{align}
\frac{\partial}{\partial x_i} \biggl | \frac{\partial \u}{\partial \u_0} \biggr | 
= \biggl | \frac{\partial \u}{\partial \u_0} \biggr | \nabla_{\u} \cdot \biggl( \frac{\partial \u}{\partial x_i} \biggr)_{rel}
= \biggl | \frac{\partial \u}{\partial \u_0} \biggr | \nabla_{\u} \cdot \vgrad_{rel,i}^\top,
\label{eq:det_rel_vel}
\end{align}
then from \eqref{eq:Re_proof0_vel_B1}
\begin{align}
\begin{split}
 d \F(\x) 
&= \sum\limits_{i=1}^3   \biggl[ \iiint\limits_{\D(\x_0)} \biggl\{
\biggl( \frac {\partial \varphi }{\partial x_i} + (\nabla_{\u} \varphi) \cdot \vgrad_{rel,i}^\top
\biggr) \biggl | \frac{\partial \u}{\partial \u_0} \biggr |  \, 
+
\varphi \biggl | \frac{\partial \u}{\partial \u_0} \biggr | \nabla_{\u} \cdot \vgrad_{rel,i}^\top
\biggr\}
dU_0 \biggr] d x_i
\\
&= \sum\limits_{i=1}^3   \biggl[ \iiint\limits_{\D(\x)} \biggl\{
 \frac {\partial \varphi }{\partial x_i} 
 + (\nabla_{\u} \varphi) \cdot \vgrad_{rel,i}^\top
+ \varphi  \nabla_{\u} \cdot \vgrad_{rel,i}^\top
\biggr\}
dU \biggr] d x_i.
\end{split}
\label{eq:Re_proof_vel_B2}
\end{align}
This gives the second form of the velocimetric-spatial Reynolds transport theorem in \eqref{eq:Re_tr_u_3D}, with the first form obtained by Gauss' divergence theorem in velocity space. $\blacksquare$

\section{\label{sect:Apx_Re_exterior} Definitions of Operators and Proof of the Multiparameter Reynolds Transport Theorem in Exterior Calculus} 
\setcounter{equation}{0}

We now prove the multiparameter Reynolds transport theorem for differential forms \eqref{eq:Re_tr_diffform}, based on multivariate extensions of exterior calculus operators and the proof of the one-parameter case 
[e.g., \onlinecite{Frankel_2013}, eqs.\ 0.49 and 4.33-4.34]. 
For this we draw on the tools of existing (one-parameter) exterior calculus, for which excellent reviews are available in a number of monographs \cite[e.g.,][]{Kobayashi_Nomizu_1963, Flanders_1963, Guggenheimer_1963, Cartan_1970, Lovelock_Rund_1989, Olver_1993, Lee_2009, Torres_del_Castillo_2012, Bachman_2012, Frankel_2013, Sjamaar_2017}. 

{\bf Proof:} Consider an $n$-dimensional differentiable manifold $M^n$, described using a patchwork of local coordinates $\X=[X_1, ..., X_n]^\top$ defined in some neighbourhood $N(\s)$ of each point $\s \in M^n$. The coordinates in $\X$ are assumed orthonormal, but need not be Cartesian. Let $\V$ be an $n \times m$ vector or tensor field on the manifold, which is parameterized by the $m$-dimensional vector of parameters $\C =[C_1,..., C_m]^\top$ (which could include time $t$). 
This field will create the $m$-parameter maximal integral curve or ``flow'' within the manifold, defined by the map
[compare \onlinecite[][chap.\ 1]{Kobayashi_Nomizu_1963}; \onlinecite[][\S3.3]{Flanders_1963}; \onlinecite[][\S 2.8]{Lee_2009}; \onlinecite[][\S 2.1]{Torres_del_Castillo_2012}; \onlinecite[][\S 1.4]{Frankel_2013}; \onlinecite[][\S 1.3]{Olver_1993}]:
\begin{align}
{\phi: M^n \times \R^m \to M^n},
\label{eq:flow1}
\end{align}
such that, respectively in vector notation (using the $\partial (\to)/\partial (\downarrow)$ convention)\footnote{For tensor fields, it may be convenient to represent the manifold using higher-order coordinates $\s \in \R^{n_1} \times ... \times \R^{n_k}$. For example, the shear stress tensor $\vec{\tau}$, represented with second order elements $\tau_{ij} \in M^{3 \times 3} \subseteq \R^3 \times \R^3$, can be used to define the third-order tensor field $\partial \vec{\tau}/\partial \C$ with elements $V_{i j c} = \partial \tau_{i j}/\partial C_c$.}:
\begin{align}
\biggl( \frac{\partial \phi (\s,\C)}{\partial \C} \biggr)^\top = \V(\s)
\label{eq:flow2}
\end{align}
or for each component $V_{jc}$, expressed in terms of the local coordinates $X_j \in \X$ and parameter components $C_c \in \C$:
\begin{align}
\frac{\partial X_j(\phi (\s, C_c))}{\partial C_c} = V_{jc}(\s)
\label{eq:flow2a}
\end{align}
The flow $\phi$ satisfies the following properties for all $\s \in M^n$ and all $\B,\C \in \R^m$  
(c.f. \cite{Kobayashi_Nomizu_1963, Olver_1993, Torres_del_Castillo_2012, Frankel_2013}):
\begin{align}
\begin{split}
\phi(\s,\0) &= \s \\
\phi(\phi(\s,\C),\B) &= \phi(\s,\B+\C).
\end{split}
\label{eq:flow3}
\end{align}
By previous custom for the one-parameter case, we write this as the bijection (diffeomorphism) 
(c.f. \cite{Kobayashi_Nomizu_1963, Olver_1993, Torres_del_Castillo_2012, Frankel_2013}):
\begin{align}
\phi^{\C}: M^n \to M^n, \hspace{10pt} \phi^{\C} (\s) = \phi(\s,\C),
\label{eq:diffeo}
\end{align}
which is therefore invertible, and operates linearly $\phi^{\C+\B} = \phi^{\C} \circ \phi^{\B}  =  \phi^{\B} \circ \phi^{\C}$.
Thus if the manifold contains an $r$-dimensional oriented compact submanifold $\Omega^r \subset M^n$, each point in the submanifold at $\C$ can be mapped from the origin at $\C=\0$ by $\Omega^r(\C) = \phi^{\C} \Omega^r (\vec{0})$, and 
vice versa $\Omega^r(\0) = \phi^{-\C} \Omega^r (\vec{C})$. 
Informally, we might describe $\Omega^r(\C)$ as a ``moving domain'' and the map $\phi^{\C}$ as a ``movement", although they each involve a transformation in the parameter vector $\C$ (such as in spatial coordinates) -- reflecting the symmetries of the vector or tensor field -- rather than necessarily in physical time. 

Now consider the $r$-form $\omega^r$, a linear function defined on the cotangent space 
of the manifold $M^n$, with $r \in \N \cup 0 = \N_0$ such that $0 \le r \le n$. This can be written as
(e.g., \onlinecite[][chap.\ 1]{Kobayashi_Nomizu_1963}; \onlinecite[][\S1.1]{Flanders_1963}; \onlinecite[][\S A.3]{Lovelock_Rund_1989})\footnote{Note that many authors adopt an implied summation convention for this and subsequent equations; we do not adopt this here.}:
\begin{align}
\omega^r 
= \sum\limits_{j_1 < ... < j_r} w_{j_1...j_r} \; dX_{j_1} \wedge ... \wedge dX_{j_r},
\label{eq:diffform_def}
\end{align}
where $w_{j_1...j_r}$ are scalars (possibly functions of $\X$), $\wedge$ is the exterior or wedge product and the $dX_{j_k}$ are an ordered selection of $r$ terms from the vector $d\X = [dX_1, ..., dX_n]^\top$,  with the sum taken over all increasing combinations of the $dX_{j_k}$. Physically, the wedge product $dX_{j_1} \wedge ... \wedge dX_{j_r}$ is the oriented volume of an infinitesimal $r$-dimensional parallelepided. Integration of $\omega^r$ over the submanifold $\Omega(\C) \subset M^n$:
\begin{align}
 W(\C) = \int\nolimits_{\Omega(\C)} \omega^r,
 \label{eq:diffform_int}
 \end{align}
therefore gives the total oriented quantity $W(\C)$ in the submanifold, as a function of its parameters $\C$. 
The $r$-form formalism thus extends standard multivariate calculus to the analysis of oriented areas and volumes on manifolds, using a patchwork of local coordinate systems.

For a smooth (infinitely differentiable) map $f: M^n \to N^\ell$ between smooth manifolds $M^n$ and $N^\ell$ (for $\ell, n \in \N$), there exists an important theorem that a smooth $r$-form $\omega^r$ on $N^\ell$ can be mapped to a smooth $r$-form $f^* \omega^r$ on $M^n$, where $f^*$ is known as the {\it pullback} 
[e.g., \onlinecite[][\S 2.7]{Frankel_2013}]. 
In consequence, assuming smoothness, the multiparametric diffeomorphism $\phi^{\C}$ defined in \eqref{eq:diffeo} can be used to define a vector pullback $\phi^{*\C}$, providing an invertible coordinate transformation between $M^n$ and itself in the $\C$ direction (with inverse $\phi^{\C}_*$, known as the {\it pushforward}). Formally, we define 
[compare \onlinecite[][\S 5.5]{Lovelock_Rund_1989}; \onlinecite[][\S 0j and \S 2.7]{Frankel_2013}; \onlinecite[][\S 3.2]{Sjamaar_2017}]:
\begin{align}
\begin{split}
\phi^{*\C} \omega^r 
&= \sum\limits_{j_1 < ... < j_r} (w_{j_1...j_r} \circ \phi^{\C}) \; d\phi^{\C}_{j_1} \wedge ... \wedge d\phi^{\C}_{j_r}
\\&= \sum\limits_{j_1 < ... < j_r} \sum\limits_{k_1 < ... < k_r} (w_{j_1...j_r} \circ \phi^{\C}) \; 
 \biggl| \frac{\partial(\phi^{\C}_{j_1}, ..., \phi^{\C}_{j_r})}{\partial(X_{k_1}, ..., X_{k_r})} \biggr|
dX_{k_1} \wedge ... \wedge dX_{k_r},
\end{split}
\label{eq:diffform_pullback}
\end{align}
where $\bigl| {\partial(\phi^{\C}_{j_1}, ..., \phi^{\C}_{j_r})}/{\partial(X_{k_1}, ..., X_{k_r})} \bigr|$ is the determinant of the Jacobian matrix between the two coordinate systems, without change of sign. We see that the pullback $\phi^{*\C}$ satisfies linearity, and enables an $r$-form at $\C$ to be mapped back to $\C=\0$, or vice versa using the pushforward $\phi_*^{\C}$.

We next consider the exterior derivative, which when applied to an $r$-form gives 
[e.g., \onlinecite[][chap.\ 1]{Kobayashi_Nomizu_1963}; \onlinecite[][\S3.2]{Flanders_1963}; \onlinecite[][\S A.3]{Lovelock_Rund_1989}; \onlinecite[][\S 8.3]{Lee_2009}; \onlinecite[][\S2.6]{Frankel_2013}]:
\begin{align}
d\omega^r 
= \sum\limits_{j_1 < ... < j_r} dw_{j_1...j_r} \wedge dX_{j_1} \wedge ... \wedge dX_{j_r}.
\label{eq:diffform_diff}
\end{align}
Since the integral in \eqref{eq:diffform_int} is a 0-form, its exterior derivative is its differential:
\begin{align}
dW(\C) 
= \sum\limits_{c=1}^m \frac{\partial W}{\partial C_c}\biggr|_{C_k \ne C_c}  \, dC_c
=  \frac{\partial W}{\partial \C} \cdot d\C
=  \frac{\partial ( \int\nolimits_{\Omega(\C)} \omega^r ) }{\partial \C} \cdot  d\C,
\label{eq:diffform_proof1}
\end{align}
which indicates the terms $C_k$, for all $k \ne c$, are held constant in each partial derivative, and which uses the standard dot product. To simplify, the variable domain of integration is first converted to a fixed domain via the pullback: 
\begin{align}
dW(\C)
= \frac{\partial ( \int\nolimits_{\Omega(\0)} \phi^{*\C} \, \omega^r )}{\partial \C} \cdot d\C
= \sum\limits_{c=1}^m \frac{\partial ( \int\nolimits_{\Omega(0)} \phi^{*C_c} \, \omega^r )} {\partial C_c} \biggr|_{C_k \ne C_c} \,  dC_c,
\label{eq:diffform_proof4}
\end{align}
using the reference position $\C=\0$, in the (relative) coordinate system chosen for $\C$.
For each component in \eqref{eq:diffform_proof4}, from the definition of the partial derivative and the linearity of the pullback \eqref{eq:diffform_pullback}: 
\begin{align}
 \begin{split}
\frac{\partial }{\partial C_c} \int\nolimits_{\Omega(0)} \phi^{*C_c} \, \omega^r 
 &= \lim\limits_{h \to 0} \frac{\int\nolimits_{\Omega(0)} \phi^{*(C_c+h)} \, \omega^r - \int\nolimits_{\Omega(0)} \phi^{*C_c} \, \omega^r} {h}
\\&
=  \lim\limits_{h \to 0} \int\nolimits_{\Omega(0)} \frac{ \phi^{*C_c} (\phi^{*h} \, \omega^r -   \omega^r)} {h}
=   \int\nolimits_{\Omega(C_c)} \biggl \{ \lim\limits_{h \to 0} \frac{ (\phi^{*h} \, \omega^r -   \omega^r)} {h} \biggr \},
\end{split}
\label{eq:diffform_proof3}
\end{align}
where the last step converts back to a variable domain using the pushforward $\phi_*^{C_c}$. The term in braces is the Lie derivative $\mathcal{L}_{\V_{\cdot c}}$ of the differential form $\omega^r$ with respect to the column vector field $\V_{\cdot c} \in \V$ associated with the flow $\phi^{C_c} \in \phi^{\C}$, based on the increment $h$ in the one-dimensional flow parameter $C_c$ 
[\onlinecite[][\S 2.2]{Torres_del_Castillo_2012}; \onlinecite[][\S 4.3a]{Frankel_2013}]. 
Taking a cue from the directional derivative (see  \ref{sect:Apx_Re_spatial}), this could equivalently be defined in terms of the pullback $\phi^{*h dC_c}$ and written as $\mathcal{L}_{\V_{\cdot c}}^{(C_c)}$, to explicitly identify the component $C_c$.  
In consequence, 
we can define an $m$-dimensional multiparameter Lie derivative of an $r$-form with respect to $\V$ over parameter $\C$ by:
\begin{align}
\begin{split}
\mathcal{L}_{\V}^{(\C)} \omega^r
 &= [ \mathcal{L}_{\V_{\cdot 1}}^{(C_1)} , ..., \mathcal{L}_{\V_{\cdot m}}^{(C_m)} ]^\top \,\omega^r
\\&=  \lim\limits_{h \to 0} \biggl [ \frac{ (\phi^{*h \, dC_1} \, \omega^r -   \omega^r)} {h}  
, ... ,
  \frac{ (\phi^{*h \, dC_m} \, \omega^r -   \omega^r)} {h} 
 \biggr]^\top
 =    \lim\limits_{h \to 0} \frac{ (\phi^{*h \, d\C} \, \omega^r -  \1_m \, \omega^r)} {h} .
\end{split}
 \label{eq:Lie_deriv_vector}
\end{align}
where $\1_m$ is an $m$-dimensional vector of 1s. 
Assembling \eqref{eq:diffform_int}-\eqref{eq:Lie_deriv_vector} then gives:
\begin{align}
dW(\C) 
= d \int\limits_{\Omega(\C)}   \omega^r 
=  \biggl[ \int\limits_{\Omega(\C)} \mathcal{L}_{\V}^{(\C)}  \omega^r \biggr] \cdot {d\C} .
\end{align}
This is the first part of \eqref{eq:Re_tr_diffform}. 

Finally, we consider the one-parameter interior product, which effects the contraction of an $r$-form to an $(r-1)$-form, given for $r>0$ by [e.g., \onlinecite[][p100]{Cartan_1970}, \onlinecite[][\S 1.5]{Olver_1993}, \onlinecite[][\S 9.2]{Lee_2009}]:
\begin{align}
\begin{split}
i_{\U} \, \omega^r  
&= \sum\limits_{j_1 < ... <j_r} \sum\limits_{k=1}^n (-1)^{k-1} U_k \, w_{j_1...j_r} dX_{j_1} \wedge ... \wedge dX_{j_{k-1}} \wedge dX_{j_{k+1}} \wedge  ... \wedge dX_{j_r}
\end{split}
\end{align}
based on components $U_k$ of a one-parameter vector field $\U$ defined on $M^n$ with implicit parameter $t$. 
This was shown by Cartan to satisfy the equation $\mathcal{L}_{\U} \omega^r =  i_{\U} \, d \omega^r + d \,  (i_{\U} \, \omega^r )$ 
[e.g., \onlinecite[][\S5.8]{Flanders_1963}, \onlinecite[][\S A.3]{Lovelock_Rund_1989}, \onlinecite[][\S 1.5]{Olver_1993},\onlinecite[][\S 8.6]{Lee_2009}, \onlinecite[][\S 4.2b]{Frankel_2013}].
By component-wise extension, it is possible to define a multiparameter interior product based on the field $\V$ with parameters $\C$:
\begin{align}
\begin{split}
i_{\V}^{(\C)} \, \omega^r  
&= \Bigl[ i_{\V_{\cdot 1}}^{(C_1)} \, \omega^r, ..., i_{\V_{\cdot m}}^{(C_m)} \, \omega^r ]^\top
\\
&= \sum\limits_{j_1 < ... <j_r} \sum\limits_{k=1}^n (-1)^{k-1} \V_{k \cdot}^\top \, w_{j_1...j_r} dX_{j_1} \wedge ... \wedge dX_{j_{k-1}} \wedge dX_{j_{k+1}} \wedge  ... \wedge dX_{j_r}
\end{split}
\label{eq:interior_vect}
\end{align}
based on row vectors $\V_{k \cdot} \in \V$. 
By construction, this satisfies a multiparameter Cartan equation:
\begin{align}
\mathcal{L}_{\V}^{(\C)} \omega^r =  i_{\V}^{(\C)}  \, d \omega^r + d \,  (i_{\V}^{(\C)} \, \omega^r )
\label{eq:multi_Cartan}
\end{align}
Using this result, and the exterior calculus expression of Stokes' theorem $\int\nolimits_{\Omega(\C)} d \omega^r = \oint\nolimits_{\partial \Omega(\C)} \omega^r$ 
[e.g., \onlinecite[][\S 5.5]{Lovelock_Rund_1989}; \onlinecite[][\S 6.2]{Bachman_2012}; \onlinecite[][\S 3.3b]{Frankel_2013}], 
we obtain the third and fourth terms in \eqref{eq:Re_tr_diffform}. $\blacksquare$

{\bf Discussion:} The above proof invokes $m$-parameter vector extensions of the ``flow'' \eqref{eq:flow1}-\eqref{eq:diffeo}, 
pullback and pushforward  \eqref{eq:diffform_pullback}, Lie derivative \eqref{eq:Lie_deriv_vector} and interior product \eqref{eq:interior_vect}, which follow naturally from their one-parameter definitions. The $r$-form, exterior derivative and dot product are unchanged. The proof also extends naturally to higher-order tensor fields and to vector- or tensor-valued differential forms, 
by component-wise application of operators, in the same manner as does the traditional Reynolds transport theorem \eqref{eq:Re_tr}.
It also can be extended to a parametric tensor $\C$, if desired, using an element-wise (Hadamard) tensor product, or alternatively by the use of trace or higher-order diagonal operators on matrix products (such as in the Frobenius inner product). 

\section{\label{sect:Apx_Re_exterior_ext} Proof of the Augmented Multiparameter Reynolds Transport Theorem in Exterior Calculus} 
\setcounter{equation}{0}

We first present a definition and several lemmas, and then the main proof.

\begin{definition}
Extending the terminology of one-parameter exterior calculus, the differential operator associated with an $n \times m$ vector or tensor field $\V$ \eqref{eq:flow2} can be defined as [c.f., \onlinecite[][\S A.1]{Lovelock_Rund_1989}, \onlinecite[][\S 1.3]{Olver_1993}, \onlinecite[][\S 2.8]{Lee_2009}, \onlinecite[][\S1.3-1.4]{Frankel_2013}]: 
\begin{align}
\Vop 
= \sum\limits_{j=1}^n \V_{j\cdot}^\top(\X) \frac{\partial}{\partial X_{j}} 
= \V(\X) \cdot \vpartial_{\X} 
\label{eq:infgen} 
\end{align}
in which $\Vop$ denotes the differential operator, and we retain the notation $\V$ for the tensor field, where $\V_{j\cdot} (\X)$ is the $j$th row of $\V$ defined at $\X$. We further define $\vpartial_{\X}= [\partial/\partial X_1, ..., \partial/\partial X_n]^\top$ as the vector partial differential operator with respect to $\X$; this notation avoids confusion with the gradient operator $\nabla_{\X}$ for a non-Cartesian coordinate system\footnote{The vector partial derivative operator has variously been denoted using index notation [e.g., \onlinecite[][\S A.1]{Lovelock_Rund_1989}, \onlinecite[][\S 1.3]{Olver_1993}, \onlinecite[][\S 2.8]{Lee_2009}], the bold operator $\vpartial/\vpartial X_j$ [\onlinecite{Frankel_2013}, \S1.3-1.4] or simply by the gradient [\onlinecite[][\S 7.1]{Guggenheimer_1963}, \onlinecite[][\S 5.1]{Bachman_2012}].}. For consistency, we also require the operator \eqref{eq:infgen} to conduct an implicit rotation, to convert the $m$-dimensional row vector into a column vector.
\end{definition}

In \eqref{eq:infgen}, the partial derivative terms can be interpreted as a basis set of tangent vectors at each point, expressed in the local coordinate system, but also act as differential operators on mathematical objects [e.g., \onlinecite[][\S 7.1]{Guggenheimer_1963}, \onlinecite[][\S A.1]{Lovelock_Rund_1989}, \onlinecite[][\S 1.3]{Olver_1993}, \onlinecite[][\S 2.8]{Lee_2009}, \onlinecite[][ \S1.3b-c]{Frankel_2013}]. 

\begin{lemma}
Application of the tensor field operator to a function $f$ gives:
\begin{align}
\Vop(\f) 
= \sum\limits_{j=1}^n \V_{j\cdot}^\top(\X) \frac{\partial \f(\X)}{\partial X_{j}} 
= \V(\X) \cdot \vpartial_{\X} \f(\X)
= D_{\V} \, \f
= \mathcal{L}_{\V}^{(\C)} f
\label{eq:infgen_f3} 
\end{align}
where $D_{\V}$ is a multiparameter directional derivative in the directions of the columns $\V_{\cdot c}$ of $\V$, i.e., with one direction for each component $C_c$, and $\mathcal{L}_{\V}^{(\C)}$ is the multiparameter Lie derivative defined in \eqref{eq:Lie_deriv_vector}.
\end{lemma}

\begin{proof}
Applying $\Vop$ \eqref{eq:infgen} to a differentiable function $\f(\X)$ gives the vector field:
\begin{align}
\Vop(\f)(\X) 
= \sum\limits_{j=1}^n \V_{j\cdot}^\top(\X) \frac{\partial \f(\X)}{\partial X_{j}} 
= \V(\X) \cdot \vpartial_{\X} \f(\X)
\label{eq:infgen_f} 
\end{align}
For a Cartesian local coordinate system $\Y \in M^n$:
\begin{align}
\Vop(\f)(\Y)
= \sum\limits_{j=1}^n \V_{j\cdot}^\top(\Y) \frac{\partial \f(\Y)}{\partial Y_{j}} 
= \V(\Y) \cdot \nabla_{\Y} \f(\Y)
= D_{\V} \, \f(\Y)
\label{eq:infgen_f_Cartesian} 
\end{align}
where $\nabla_{\Y}$ is the gradient with respect to $\Y$. The multiparameter directional derivative in \eqref{eq:infgen_f_Cartesian} is obtained by assembling its vector components. By coordinate transformation to any other orthogonal coordinates $\X$, using the definition of $\V$ in \eqref{eq:flow2}-\eqref{eq:flow2a}:
\begin{align}
\begin{split}
\Vop(\f)(\Y)
&= D_{\V} \, \f(\Y)
= \V(\Y)^{\top} \, \nabla_{\Y} \f(\Y)
= \V(\Y)^{\top} \; \mathcal{J} \mathcal{J}^{-1} \; \nabla_{\Y} \f(\Y)
\\&=  \biggl (\frac{\partial \Y}{\partial \C} \biggr)^\top \, \frac{\partial \X}{\partial \Y} \,  \frac{\partial \Y}{\partial \X} \, \nabla_{\Y} \f(\Y)
=  \biggl\{ \biggl (\frac{\partial \Y}{\partial \C} \biggr)^\top \, \frac{\partial \X}{\partial \Y} \biggr\} \biggl\{ \frac{\partial \Y}{\partial \X} \, \nabla_{\Y} \f(\Y) \biggr\}
\\&= \biggl(\frac{\partial \X}{\partial \C} \biggr)^\top \vpartial_{\X} \f(\X) 
= \V(\X)^\top \vpartial_{\X} \f(\X) 
=\Vop(\f)(\X)
\end{split}
\label{eq:infgen_f2} 
\end{align}
where $\mathcal{J} = \vpartial_{\Y} \X$ is the Jacobian of $\X$ with respect to $\Y$, and $\mathcal{J}^{-1}$ as its inverse. We see that applying the tensor field operator to a function, or equivalently the multiparameter directional derivative, is independent of the coordinate system used [c.f., \onlinecite[][\S A.3]{Lovelock_Rund_1989}]. 


Now consider the multiparameter Lie derivative of a function. For each term in \eqref{eq:Lie_deriv_vector}:
\begin{align}
\begin{split}
\mathcal{L}_{\V_{\cdot c}}^{(C_c)} f
 &=    \lim\limits_{h \to 0} \frac{ \phi^{*h \, dC_c} \,f(\X) -   f(\X)} {h} 
 \\&=    \lim\limits_{h \to 0} \frac{ f(\phi^{h \, dC_c}(\X)) -  f(\X)} {h} 
 \\&=    \lim\limits_{h \to 0} \frac{ f(\X + h \frac{\partial \X}{\partial_{C_c}} dC_c) -   f(\X)} {h} 
 \\&=    \lim\limits_{h \to 0} \frac{ f(\X + h \V_{\cdot c} dC_c) -   f(\X)} {h} 
 \\&=  D_{\V_{\cdot c}} \, \f
\end{split}
\label{eq:Lie_deriv_fn}
\end{align}
using the definition of the pullback \eqref{eq:diffform_pullback} and coordinate transformation of the increment $h\, dC_c$, where we recognise $dC_c$ as a scalar quantity (see \eqref{eq:Re_proof1_dir_vel} and the comments after \eqref{eq:diffform_proof3}).  Assembling \eqref{eq:Lie_deriv_fn} into the vector Lie derivative and uniting with \eqref{eq:infgen_f2} gives \eqref{eq:infgen_f3}. 
\end{proof}

Eq.\ \eqref{eq:infgen_f3} extends the known results of one-parameter exterior calculus [\onlinecite[][\S 7.1]{Guggenheimer_1963}, \onlinecite[][\S A.1, A.3]{Lovelock_Rund_1989}, \onlinecite[][\S 1.3-1.5]{Olver_1993}, \onlinecite[][\S 2.8]{Lee_2009}, \onlinecite[][\S1.3b-c, 1.4a, 4.2]{Frankel_2013}], e.g., for the velocity vector field $\u(\x,t)$ with coordinates $\x$ and parameter $t$, $\mathcal{L}_{\u} \f=D_{\u} \, \f$. 

\begin{lemma}
The multiparameter Lie derivative \eqref{eq:Lie_deriv_vector} exhibits the properties of termwise application (derivation), commutativity with the exterior derivative, and linearity with respect to tensor fields, respectively:
\begin{align}
\mathcal{L}_{\V}^{(\C)} (\omega^r + \tau^q) &= (\mathcal{L}_{\V}^{(\C)} \omega^r) \wedge \tau^q + \omega^r \wedge (\mathcal{L}_{\V}^{(\C)} \tau^q)
\label{eq:Lie_termwise}
\\
\mathcal{L}_{\V}^{(\C)} d &= d \mathcal{L}_{\V}^{(\C)}
\label{eq:Lie_extd}
\\
\mathcal{L}_{\V+\W}^{(\C)} \omega^r&= \mathcal{L}_{\V}^{(\C)} \omega^r + \mathcal{L}_{\W}^{(\C)} \omega^r
\label{eq:Lie_linearity}
\end{align}
where $\tau^q$ is a $q$-form defined on $M^n$ with $q \in \N_0$ and $0 \le q \le n$, and $\V$ and $\W$ are two tensor fields with $m$ columns but not necessary the same number of rows. 
\end{lemma}

\begin{proof}
The proofs of \eqref{eq:Lie_termwise} and \eqref{eq:Lie_extd} follow columnwise from their one-parameter counterparts [c.f., \onlinecite[][\S A.3]{Lovelock_Rund_1989},  \onlinecite[][\S 1.5]{Olver_1993}, \onlinecite[][\S 8.6]{Lee_2009}, \onlinecite[][\S 4.2a]{Frankel_2013}]. 
To prove  \eqref{eq:Lie_linearity}, consider the $n_1 \times m$ tensor field $\V$ and $n_2 \times m$ tensor field $\W$ defined on the same manifold $M^n$ with $1 \le n_1, n_2 \le n$, respectively with flows $\phi(\s,\C)$ and $\psi(\s,\C)$ defined by \eqref{eq:flow2}-\eqref{eq:flow2a}. Let the tensor field $\V$ have local $n_1$-dimensional coordinates $\X$ and operator $\Vop$, and let $\W$ have the local $n_2$-dimensional coordinates $\Y$ and operator $\Wop$.
From the definition \eqref{eq:Lie_deriv_vector} of the multiparameter Lie derivative:
\begin{align}
\begin{split}
\mathcal{L}_{\V}^{(\C)} \omega^r + \mathcal{L}_{\W}^{(\C)} \omega^r 
 &=    \lim\limits_{h \to 0} \frac{ \phi^{*h \, d\C} \, \omega^r -  \1_m \, \omega^r} {h} 
 + \lim\limits_{h \to 0} \frac{ \psi^{*h \, d\C} \, \omega^r -  \1_m \, \omega^r} {h} 
 \\&= \lim\limits_{h \to 0} \frac{ \phi^{*h \, d\C} \, \omega^r + \psi^{*h \, d\C} \, \omega^r -  \2_m \, \omega^r} {h} 
 \\&= \lim\limits_{h \to 0} \frac{ (\phi + \psi)^{*h \, d\C} \, \omega^r -  \2_m \, \omega^r} {h} 
 \\&= \lim\limits_{h' \to 0} \frac{ (\phi + \psi)^{*h' \, d\C} \, \omega^r -  \1_m \, \omega^r} {h'} 
 \\&= \mathcal{L}_{\V+\W}^{(\C)} \, \omega^r
\end{split}
 \label{eq:Lie_linearity2}
\end{align}
by linearity of the pullback \eqref{eq:diffform_pullback} and redefinition of the distance $h=2h'$, where $\2_m$ denotes an $m$-dimensional vector of 2s. Note that the sum in the amalgamated Lie derivative is defined in terms of its differential operators, i.e., from \eqref{eq:infgen}:
\begin{align}
\mathcal{L}_{\V+\W}^{(\C)} := \mathcal{L}_{\Vop+\Wop}^{(\C)} = \mathcal{L}_{\V(\X) \cdot  \vpartial_{\X}+\W(\Y) \cdot \vpartial_{\Y}}^{(\C)}
 \label{eq:Lie_linearity3}
\end{align}
showing that the operators are of compatible dimension.
\end{proof}

Eqs.\ \eqref{eq:Lie_linearity2}-\eqref{eq:Lie_linearity3} extend a known result of one-parameter exterior calculus [\onlinecite[][\S 7.1]{Guggenheimer_1963}, \onlinecite[][\S 4.3b]{Frankel_2013}], with greater attention to the handling of vectors or tensors of different length. 

\vspace{10pt}

{\bf Main Proof:} We now consider the proof of \eqref{eq:Re_tr_diffform_aug}, by applying \eqref{eq:Re_tr_diffform} to a parameter-dependent vector or tensor field $\V(\C)$. Extending the analysis given in [\onlinecite[][\S 4.3b]{Frankel_2013}], we embed the manifold in the higher-order differentiable manifold $M^n \times \R^m$, in which $M^n$ is augmented with the domain of $\C$. This invokes the augmented local coordinates $\hat{\X}=[\X,\C]^\top$, again assumed orthonormal but not necessarily Cartesian. This creates the $(n+m) \times m$ vector or tensor field $\V \comp \C$, the maximal integral curves of which can be expressed by the map:
\begin{align}
{\hat{\phi}: M^n \times \R^m \times \R^m \to M^n \times \R^m},
\label{eq:flow1_aug}
\end{align}
such that, for the augmented position $\hat{\s} \in M^n \times \R^m$, augmented local coordinates $\hat{X}_j \in \hat{\X}$ and parameter components $C_c \in \C$:
\begin{align}
\biggl( \frac{\partial \hat{\phi} (\hat{\s},\C)}{\partial \C} \biggr)^\top = \V \comp \C(\hat{\s})
\hspace{15pt} \text{and} \hspace{15pt} 
\frac{\partial \hat{X}_j(\phi (\hat{\s}, C_c))}{\partial C_c} = [V \comp \C]_{jc}(\hat{\s})
\label{eq:flow2_aug}
\end{align}
The map $\hat{\phi}$ satisfies the same linearity properties as the ``flow'' $\phi$ for $\C$-independent systems ( \ref{sect:Apx_Re_exterior}), and so can be applied to the submanifold $\Omega^r$. 
By previous custom, we write the map as the diffeomorphism:
\begin{align}
\hat{\phi}^{\C}: M^n \times \R^m \to M^n \times \R^m, \hspace{10pt} \hat{\phi}^{\C} (\hat{\s}) = \hat{\phi}(\hat{\s},\C).
\label{eq:diffeo_aug}
\end{align}
This allows the definition of the augmented vector pullback $\hat{\phi}^{*\C}$ and pushforward $\hat{\phi}^{\C}_*$, enabling invertible coordinate transformations within $M^n \times \R^m$ parameterized by $\C$, which can be projected into $M^n$.

Examining \eqref{eq:flow2_aug}, $\V \comp \C$ consists of the elements $\partial X_j/\partial C_c=V_{jc}$ based on local coordinates $X_j$ in the top $n$ rows, and $\partial C_k/\partial C_c$ below, giving: 
\begin{align}
\V \comp \C = \begin{bmatrix} \V \\ \I_m \end{bmatrix}
\end{align}
where $\I_m$ is the identity matrix of size $m$. From the definition \eqref{eq:infgen}, the tensor field can be written as the augmented differential operator:
\begin{align}
\hat{\Vop} 
= (\V \comp \C) \cdot \vpartial_{\X,\C} 
= \begin{bmatrix} \V \\ \I_m \end{bmatrix} \cdot \begin{bmatrix} \vpartial_{\X} \\ \vpartial_{\C} \end{bmatrix}
=  \V \cdot \vpartial_{\X} + \I_m \cdot \vpartial_{\C}
=  \Vop + \vpartial_{\C}
\label{eq:infgen2} 
\end{align}
As an example, for the velocity vector field $\u(\x,t)$ with coordinates $\x$ and parameter $t$, \eqref{eq:infgen2} reduces to the operator $\hat{\nu} = \u \cdot \vpartial_{\x} + \partial/\partial t$ [e.g., \onlinecite{Frankel_2013}, \S4.3]. 

We now apply \eqref{eq:Re_tr_diffform} to the augmented system, noting that $\omega$ remains an $r$-form in $\Omega(\C)$:
\begin{align}
\begin{split}
\hat{d} \int\limits_{\Omega(\C)}   \omega^r 
&=  \biggl[ \int\limits_{\Omega(\C)} \mathcal{L}_{\V \comp \C}^{(\C)}  \omega^r \biggr] \cdot {d\C} 
\end{split}
\label{eq:Re_tr_diffform_aug1}
\end{align}
where $\hat{d}$ is the exterior derivative based on the augmented coordinates $\hat{\X}$. 
To reduce \eqref{eq:Re_tr_diffform_aug1}, we rewrite the Lie derivative in operator notation \eqref{eq:infgen2}, expand using \eqref{eq:Lie_linearity} and convert back:
\begin{align}
\mathcal{L}_{\V \comp \C}^{(\C)}  \omega^r 
= \mathcal{L}_{\hat{\Vop}}^{(\C)}  \omega^r 
= \mathcal{L}_{\Vop + \vpartial_{\C}}^{(\C)}  \omega^r  
= \mathcal{L}_{\Vop}^{(\C)}  \omega^r  + \mathcal{L}_{\vpartial_{\C}}^{(\C)}  \omega^r  
= \mathcal{L}_{\V}^{(\C)}  \omega^r  + \mathcal{L}_{\I_m}^{(\C)}  \omega^r  
\end{align}
Now from \eqref{eq:Lie_extd} and \eqref{eq:infgen_f3}:
\begin{align}
\mathcal{L}_{\V}^{(\C)} d X_j
= d \mathcal{L}_{\V}^{(\C)} X_j
= d (  \V(\X) \cdot \vpartial_{\X} X_j)
= d \V_{j \cdot}^\top
\end{align}
using $\partial X_j/\partial X_k =1$ if $j=k$ and $0$ if $j \ne k$, hence:
\begin{align}
\mathcal{L}_{\I_m}^{(\C)} d X_j
= d \mathcal{L}_{\I_m}^{(\C)} X_j
= d (\I_m)_{j \cdot}^\top
= \0_m
\label{eq:L_I_dX}
\end{align}
where the last step gives a zero vector of dimension $m$. Applying the identity Lie derivative to the $r$-form $\omega^r$ in \eqref{eq:diffform_def} then gives:
\begin{align}
\begin{split}
\mathcal{L}_{\I_m}^{(\C)} \omega^r 
&= \sum\limits_{j_1 < ... < j_r} \mathcal{L}_{\I_m}^{(\C)} \bigl( w_{j_1...j_r} \; dX_{j_1} \wedge ... \wedge dX_{j_r} \bigr)
\\
&= \sum\limits_{j_1 < ... < j_r} 
 (\mathcal{L}_{\I_m}^{(\C)}  w_{j_1...j_r} ) \; dX_{j_1} \wedge ... \wedge dX_{j_r} 
+ w_{j_1...j_r} \; (\mathcal{L}_{\I_m}^{(\C)} dX_{j_1}) \wedge ... \wedge dX_{j_r} 
+ ... 
\\& \quad +  w_{j_1...j_r} \; dX_{j_1} \wedge ... \wedge (\mathcal{L}_{\I_m}^{(\C)} dX_{j_r} )
\\&=\sum\limits_{j_1 < ... < j_r} 
 (\mathcal{L}_{\I_m}^{(\C)}  w_{j_1...j_r}) \; dX_{j_1} \wedge ... \wedge dX_{j_r} 
\\&=\sum\limits_{j_1 < ... < j_r} 
(\I_m \cdot \vpartial_{\C}  w_{j_1...j_r}) \; dX_{j_1} \wedge ... \wedge dX_{j_r} 
\\&=\sum\limits_{j_1 < ... < j_r} 
(\vpartial_{\C}  w_{j_1...j_r}) \; dX_{j_1} \wedge ... \wedge dX_{j_r} 
\\& = \vpartial_{\C}  \omega^r 
\end{split}
\raisetag{120pt}
\end{align}
where the second line follows from \eqref{eq:Lie_termwise}, the third line follows from \eqref{eq:L_I_dX}, and the fourth line follows from \eqref{eq:infgen_f3}, using the fact that each $w_{j_1...j_r}$ is a function (0-form) and $\mathcal{L}_{\I_m}^{(\C)}$ invokes the operator $\vpartial_{\C}$. The last line follows by amalgamation into the $r$-form, using $\vpartial_{\C} dX_{j_i}=\0_m$ for all $j$. 
In consequence, \eqref{eq:Re_tr_diffform_aug1} simplifies to:
\begin{align}
\begin{split}
\hat{d} \int\limits_{\Omega(\C)}   \omega^r 
&=  \biggl[ \int\limits_{\Omega(\C)} \mathcal{L}_{\V \comp \C}^{(\C)}  \omega^r \biggr] \cdot {d\C} 
\end{split}
\label{eq:Re_tr_diffform_aug2}
\end{align}
\begin{align*}
\begin{split}
&=  \biggl[ \int\limits_{\Omega(\C)} (\mathcal{L}_{\I_m}^{(\C)}  \omega^r +  \mathcal{L}_{\V}^{(\C)}  \omega^r ) \biggr] \cdot {d\C} 
\\& = \biggl[ \int\limits_{\Omega(\C)} ( \vpartial_{\C}  \omega^r +  \mathcal{L}_{\V}^{(\C)}  \omega^r ) \biggr] \cdot {d\C} 
\\& = \biggl[ \int\limits_{\Omega(\C)} ( \vpartial_{\C}  \omega^r + i_{\V}^{(\C)}  \, d \omega^r + d  ( i_{\V}^{(\C)} \, \omega^r  ) \biggr] \cdot {d\C} 
\end{split}
\end{align*}
where the last line follows from the multiparameter Cartan relation \eqref{eq:multi_Cartan}, in which $d$ and $i_{\V}^{(\C)}$ are based on the standard local coordinates $\X$. Eq.\ \eqref{eq:Re_tr_diffform_aug2} connects the first, second and fourth parts of \eqref{eq:Re_tr_diffform_aug}. The third part of \eqref{eq:Re_tr_diffform_aug}, containing the surface integral term, follows from Stokes' theorem.  $\blacksquare$

We note that if $\C$ is expressed in Cartesian coordinates, the first term in the integrand of \eqref{eq:Re_tr_diffform_aug2} can be written as $\nabla_{\C}  \omega^r$. 
The above analysis extends the proof of \eqref{eq:Re_tr_diffform_aug2} for one-parameter systems $\C=t$ given by [\onlinecite[][\S 4.3b]{Frankel_2013}]. The same result appears to have been first reported by Flanders, using a different proof based on $r$-chains and an augmented pullback operator [\onlinecite[][\S8]{Flanders_1973}].

\section{\label{sect:Apx_prob_forms} Probability $r$-forms} 
\setcounter{equation}{0}

There is a complication in the definition of probability $r$-forms, due to the question of orientation 
[e.g.\ \onlinecite[][\S 11.4]{Folland_1999}]. 
This arises from the contradiction between the measure-theoretic definition of a probability density, which is independent of the direction of integration, and the oriented volumes and surfaces encountered in exterior calculus.
To address this, we first define a probability $r$-form by:
\begin{align}
\rho^r 
= \sum\limits_{j_1 < ... < j_r} \text{$\thn$}_{j_1...j_r} \; dX_{j_1} \wedge ... \wedge dX_{j_r},
\label{eq:prob_form_expansion}
\end{align}
where $\text{$\thn$}_{j_1...j_r}$ are scalars and the $dX_{j_k}$ are an ordered selection of $r$ vectors from $[dX_1, ..., dX_n]^\top$. This definition is made subject to local and global constraints, respectively:
\begin{align}
\begin{split}
\rho^r &\ge 0, \hspace{10pt} \forall \s \in M^n \\
\int\nolimits_{\Omega(\C)} \rho^r &=1,  \hspace{10pt} \forall \Omega(\C) \in M^n.
\end{split}
\label{eq:prob_form_constraints}
\end{align}
To satisfy these constraints, we define \eqref{eq:prob_form_expansion}-\eqref{eq:prob_form_constraints} only for an oriented compact submanifold $\Omega(\C)$ within an orientable manifold $M^n$, and preclude non-orientable manifolds \cite{Folland_1999}. 
Furthermore, the choices of the $\thn_{j_1...j_r}$ terms and/or the combinations of $dX_{j_k}$ need to be restricted with respect to the orientation of the submanifold $\Omega(\C)$ to satisfy the constraints. 
The $\thn_{j_1...j_r}$ terms can then be interpreted as connected segments or portions of a joint-conditional pdf $\hat{p}(\s |\C)$ defined over all points $\s \in \Omega(\C)$ in the submanifold, using a local coordinate system $\X(\s)$, subject to the conditions $\C$.

The nonnegativity constraint in \eqref{eq:prob_form_constraints} can be achieved in several ways: the simplest method is to take $\rho^r$ as the absolute and normalized value of some $r$-form $\upsilon^r$ defined over the submanifold.  A weaker method would be to impose the equivalence $\sign(\text{$\thn$}_{j_1...j_r}) = \sign(dX_{j_1} \wedge ... \wedge dX_{j_r})$, ensuring non-negative terms in the sum. An even weaker method would be to allow negative local terms $\thn_{j_1...j_r} < 0$ and oriented volume elements $dX_{j_1} \wedge ... \wedge dX_{j_r} < 0$, so long as the constraints \eqref{eq:prob_form_constraints} are satisfied in the sum \eqref{eq:prob_form_expansion}.

\section{\label{sect:Apx_Re_param} Proof of the Generalized Reynolds Transport Theorem in Vector Calculus} 
\setcounter{equation}{0}

To prove \eqref{eq:Re_tr_vect_calc_gen}, consider the augmented Reynolds transport theorem \eqref{eq:Re_tr_diffform_aug} with global Cartesian coordinates $\X$ and parameters $\C$ defined on the space $M \subseteq \R^n$, for which the tensor field $\V$ is a function of $\X$ and $\C$. This theorem is applied to the {\it top form} $\mu^n = \sum\nolimits_{j_1,...,j_n} w_{j_1...j_n} $ $dX_{1} \wedge ... \wedge dX_{n}$ defined on $M$. Since $M$ is orientable, $\mu^n$ is also a {\it volume form}, so we can set $\mu^n = \psi \, dX_{1} \wedge ... \wedge dX_{n}$ based on the non-vanishing density field $\psi(\X, \C) = \sum\nolimits_{j_1,...,j_n} w_{j_1...j_n}$ \cite[][\S 8.7]{Lee_2009}. We further assume that $\psi$ is continuous and continuously differentiable with respect to $\X$ and $\C$ throughout the domain $\Omega(\C) \subset M$, for all coordinates up to its boundary and all parameter values considered. 

Examining the left-hand side of \eqref{eq:Re_tr_diffform_aug}, since the integral is a function (0-form) we see that
$\hat{d} \int\nolimits_{\Omega(\C)}  \mu^n = d \int\nolimits_{\Omega(\C)}  \mu^n$, 
where $d$ is the differential. Now consider each term in the last integrand on the right-hand side of \eqref{eq:Re_tr_diffform_aug} applied to $\mu^n$. Firstly, for Cartesian parameters $\C$, the first term reduces to $\vpartial_{\C}  \mu^r = \nabla_{\C}  \mu^r$.  Secondly, from the definition of the exterior derivative \eqref{eq:diffform_diff}:
\begin{align}
\begin{split}
d\mu^n 
&= d\psi \wedge dX_{1} \wedge ... \wedge dX_{n}
=\sum\limits_{k=1}^n \frac{\partial \psi}{\partial X_k} dX_k \wedge dX_{1} \wedge ... \wedge dX_{n}
=0
\end{split}
\label{eq:diffform_diff2}
\end{align}
since every term contains $dX_k \wedge dX_k=0$ for some $k \in \{1,...,n\}$, by virtue of being a top form. In consequence, the second term $i_{\V}^{(\C)}  \, {d} \mu^n$ in \eqref{eq:Re_tr_diffform_aug} reduces to the zero vector $\0_m$. Thirdly, examining the last term $d i_{\V}^{(\C)} \, \mu^n$ in \eqref{eq:Re_tr_diffform_aug}, using \eqref{eq:interior_vect}, \eqref{eq:diffform_diff} and $dX_i \wedge dX_i =0$ we find that [c.f., \onlinecite[][\S 9.2]{Lee_2009}]:
\begin{align}
\begin{split}
d i_{\V}^{(\C)} \, \mu^n
&= d i_{\V}^{(\C)} \bigl( \psi \, dX_{1} \wedge ... \wedge dX_{n} \bigr)
\\
&= d \sum\limits_{k=1}^n (-1)^{k-1} \, \psi \V_{k \cdot}^\top \;  dX_{1} \wedge ... \wedge dX_{k-1} \wedge dX_{k+1} \wedge  ... \wedge dX_{n}
\\
&= \sum\limits_{k=1}^n (-1)^{k-1} \, d(\psi \V_{k \cdot}^\top) \wedge  dX_{1} \wedge ... \wedge dX_{k-1} \wedge dX_{k+1} \wedge  ... \wedge dX_{n}
\\
&= \sum\limits_{k=1}^n (-1)^{k-1} \biggl( \sum\limits_{j=1}^n \frac{\partial}{\partial X_j} (\psi \V_{k \cdot}^\top) \, dX_j  \biggr)  \wedge dX_{1} \wedge ... \wedge dX_{k-1} \wedge dX_{k+1} \wedge  ... \wedge dX_{n}
\\
&= \sum\limits_{k=1}^n (-1)^{k-1}  \frac{\partial}{\partial X_k} (\psi \V_{k \cdot}^\top) \, dX_k \wedge  dX_{1} \wedge ... \wedge dX_{k-1} \wedge dX_{k+1} \wedge  ... \wedge dX_{n}
\\
&= \sum\limits_{k=1}^n   \frac{\partial}{\partial X_k} (\psi \V_{k \cdot}^\top)  \;  dX_{1} \wedge  ... \wedge dX_{n}
\\
&= \nabla_{\X} \cdot  (\psi \V)  \;  dX_{1} \wedge  ... \wedge dX_{n}
\end{split}
\raisetag{50pt}
\end{align}
where we again define $ \nabla_{\X} \cdot (  \psi \, \V  ) =  [\nabla_{\X}^\top ( \psi \, \V )]^\top$.  Assembling these results into \eqref{eq:Re_tr_diffform_aug}, and recognising that integration over $dX_{1} \wedge  ... \wedge dX_{n}$ is equivalent to integration over $dX_{1} ... dX_{n} = d^n \X$, we establish the equivalence of the first and last terms of \eqref{eq:Re_tr_vect_calc_gen}. The middle term in \eqref{eq:Re_tr_vect_calc_gen}, containing a surface integral, is obtained from the last term by the Gauss-Ostrogradsky divergence theorem $\blacksquare$.



\end{document}